\documentclass[11pt]{amsart}
\usepackage{amsfonts}
\usepackage{amsmath}
\usepackage{amscd}
\usepackage{amsthm}
\usepackage{enumerate}

\usepackage{graphicx}

\usepackage{color}

\usepackage{hyperref}

\usepackage{mathtools}                      
\mathtoolsset{showonlyrefs=true}            

\usepackage[margin=1in]{geometry}

\newtheorem{theorem}{Theorem} 
\newtheorem*{theorem*}{Theorem}

\newtheorem{definition}[theorem]{Definition}
\newtheorem{proposition}[theorem]{Proposition}

\theoremstyle{remark}
\newtheorem{rmk}[theorem]{Remark}
\newtheorem{asmp}[theorem]{Assumption}

\newcommand{\rr}{\mathbb{R}}

\newcommand{\ep}{\hfill \ensuremath{\Box}}
\newcommand{\eq}{\begin{equation}}
\newcommand{\en}{\end{equation}}

\newcommand{\ev}{\mathbb{E}}

\numberwithin{equation}{section} \numberwithin{theorem}{section}

\title[Forward performance processes]{Asymptotic Analysis of Forward performance processes in incomplete markets and their ill-posed HJB equations}

\author{Mykhaylo Shkolnikov}
\address{Department of Operations Research and Financial Engineering, Princeton University}
\email{mshkolni@gmail.com}

\author{Ronnie Sircar}
\address{Department of Operations Research and Financial Engineering, Princeton University}
\email{sircar@princeton.edu}

\author{Thaleia Zariphopoulou}
\address{Departments of Mathematics \& IROM, McCombs School of Business, University of Texas, Austin and the Oxford-Man Institute, University of Oxford}
\email{zariphop@math.utexas.edu}

\begin{document}

\begin{abstract}
We consider the problem of optimal portfolio selection under forward investment performance criteria in an incomplete market. The dynamics of the prices of the traded assets depend on a pair of stochastic factors, namely, a slow factor (e.g. a macroeconomic indicator) and a fast factor (e.g. stochastic volatility). We analyze the associated forward performance SPDE and provide explicit formulae for the leading order and first order correction terms for the forward investment process and the optimal feedback portfolios. They both depend on the investor's initial preferences and the dynamically changing investment opportunities. The leading order terms resemble their time-monotone counterparts, but with the appropriate stochastic time changes resulting from averaging phenomena. The first-order terms compile the reaction of the investor to both the changes in the market input and his recent performance. Our analysis is based on an expansion of the underlying ill-posed HJB equation, and it is justified by means of an appropriate remainder estimate.
\end{abstract}

\maketitle


\section{Introduction}

This paper analyzes the optimal portfolio selection problem under forward investment criteria in incomplete markets. Incompleteness stems from the presence of imperfectly correlated stochastic factors that affect the dynamics of the traded assets. Such factors have been widely used in the literature and model an array of market inputs, like, among others, stochastic volatility, stochastic interest rates, predictability of asset returns, and various macroeconomic indicators. Herein, we consider a pair of such factors, which are taken, however, to move at different time scales. 

\medskip

The mathematical formulation of the problem of optimal investment in continuous time was pioneered by Merton in \cite{Me1}, \cite{Me2} and is usually referred to as the Merton problem. In the classical Merton problem the investor faces a complete market and seeks an investment portfolio that optimizes her expected utility from wealth acquired in the investment process. Therein, the investor's utility function (or, equivalently, her preferences) is determined ex ante and does not change over time. The Merton problem has been studied in a variety of frameworks and we refer to the books \cite{Du}, \cite{KS} for excellent accounts of the classical results. 
However, the setup of the Merton problem has two inherent drawbacks: 1) the investor has to decide about her utility function ex ante and cannot adapt it to market observations; 2) the investments over different time horizons are typically inconsistent with each other: for example, for $0<T_1<T_2$ the solution to the investment problem for the time period $[0,T_1]$ is in general not the restriction to $[0,T_1]$ of the solution to the investment problem for the time period $[0,T_2]$ with the $T_1$ utility function also imposed at $T_2$. In particular, it is not clear how to invest in a time-consistent fashion over rolling horizons (deterministic or random).

\medskip

{\em Forward investment performance} criteria were introduced and developed in \cite{MZ1} and \cite{MZ2}, and provide a complementary setting to the traditional expected utility framework. They allow for dynamic adaptation of the investor's risk preferences given how the market conditions change and, also, take into account the updated performance of the implemented strategies. 
The forward performance process $U(t,\cdot)$, $t\ge0$ is a stochastic process adapted to the filtration of the investor with the properties that with probability one all functions $x\mapsto U(t,x)$ are increasing and concave (and, thus, can serve as utility functions); for every self-financing strategy $\pi$ and the corresponding portfolio value process $X^\pi$ the process $U(t,X^\pi(t))$, $t\ge0$ is a supermartingale in the filtration of the investor; and there exists a self-financing strategy $\pi^*$ such that the process $U(t,X^{\pi^*}(t))$, $t\ge0$ is a martingale in the filtration of the investor. 
The pair $(U,\pi^*)$ encodes how the preferences of the investor and her optimal investment decisions jointly evolve in time from the given $U(0,\cdot)$. For a more detailed description of these criteria, further motivation and construction of concrete examples, we refer the reader, among others, to \cite{ElM}, \cite{KOZ}, \cite{MZ5}, \cite{MZ6}, \cite{NT}, and \cite{NZ1}.

\medskip

The specification of the initial datum $U(0,\cdot)$ is a central issue and a topic of current research in the forward portfolio selection approach. From the theoretical point of view, the main question is the characterization of the set of admissible initial data that give a well-defined solution. This was addressed in \cite{MZ4} for the case of time-monotone forward processes and we refer to these results repeatedly herein.

\medskip

From the practical point of view, the question is how to use and translate
investor's targets - for example, desired upcoming investment performance
or personal market views, etc. - to a mathematical input. Questions of
this nature were studied in \cite{MZ4} (section 5.2) where it is shown how an investment target on the upcoming average return can be used to infer the
initial utility input. For example, one can think of a client presenting a fund manager with the desired investment target (e.g. 5\% above the S\&P 500 performance) and a band around the investment target (e.g. 4-6\% above the S\&P 500 performance), which give an indication about the client's initial utility function $U(0,\cdot)$. Then, the fund manager's problem is to find a pair $(U,\pi^*)$ with the given $U(0,\cdot)$. Therefore, the question of finding large classes of forward performance processes $U$ and the corresponding optimal portfolios $\pi^*$ is of great importance. 

\medskip

Assuming the filtration of the investor to be generated by a Brownian motion and her forward performance process $U(t,x)$ to be an It\^o process in $t$ and twice continuously differentiable in $x$, one can show (see \cite{MZ5} and \cite{NT} for more details) that $U$ is a solution of the fully nonlinear stochastic partial differential equation (SPDE)
\begin{equation}\label{SPDE0}
\mathrm{d}U(t,x)=\frac{1}{2}\,\frac{\big\|U_x(t,x)\,\lambda(t)+\sigma(t)\,\sigma^{-1}(t)\,a^W_x(t,x)\big\|^2}
{U_{xx}(t,x)}\,\mathrm{d}t+a(t,x)^T\,\mathrm{d}\mathcal{W}(t).
\end{equation}        
Here $\mathcal{W}=(W,\tilde{B})$ is a standard Brownian motion that generates the filtration of the investor; $W$ is a Brownian motion to whose filtration the asset prices are adapted; $\sigma$ is the corresponding volatility matrix of the asset prices and $\sigma^{-1}$ is its Moore-Penrose pseudoinverse; $\lambda$ is the market price of risk; $a=(a^W,a^{\tilde{B}})$ is a suitable stochastic process adapted to the filtration of the investor; and the superscript $T$ denotes transpose. 

\medskip

The forward SPDE \eqref{SPDE0} provides the analogue of the Hamilton-Jacobi-Bellman (HJB) equation that is associated with the classical optimization problems of expected utility from terminal wealth.
As in the traditional setting, it is fully nonlinear and possibly degenerate. There are, however, fundamental differences between \eqref{SPDE0} and its classical counterpart. Firstly, \eqref{SPDE0} is posed forward in time, which makes the problem ill-posed. Secondly, the forward volatility process $a(t,x)$ is up to the investor to choose, in contrast to the classical case, where it is the mere outcome of the It\^o decomposition of the value function process. The specification of the correct class of forward volatility processes is a very challenging problem, which remains open.

\medskip

So far three classes of forward performance processes have been exhibited in the literature: 1) time-monotone forward performance processes, that is: forward performance processes which are of finite variation in the time variable (see \cite{MZ4} for more details); 2) homothetic forward performance processes, that is: forward performance processes whose dependence on the investor's wealth $x$ is of power form (see \cite{NZ1} and \cite{NT} for more details); 3) forward performance processes of factor form in complete markets (see \cite{NT} for more details). These three types of forward performance processes result from significant simplifications of the SPDE \eqref{SPDE0} in certain special cases: time-monotone forward performance processes are obtained by setting $a\equiv0$ and solving the resulting partial differential equation (PDE); homothetic forward performance processes result from choosing $U$ to be of product form with power function dependence on $x$; and forward performance processes of factor form in complete markets are derived using a reduction of the SPDE \eqref{SPDE0} to a Hamilton-Jacobi-Bellman (HJB) equation that can be linearized in the complete market framework using the Fenchel-Legendre transform.  

\medskip

As in \cite{MZ5} and \cite{NT} we consider forward performance processes of factor form (see Subsection \ref{sub_setting} for the exact details), so that \eqref{SPDE0} can be reduced to an (ill-posed) HJB equation. However, we consider the \textit{incomplete market} case in which the HJB equation cannot be linearized by a simple transformation. Nonetheless, we are able to find \textit{explicit formulas} for the leading order and first-order correction terms of the solution to such an HJB equation. These yield the leading order and first-order correction terms of the corresponding pair $(U,\pi^*)$. We complement these results by an appropriate estimate of the remainder term. 

\medskip

A similar expansion for the classical Merton problem in an incomplete market was given in \cite{FSZ}. In contrast to the Merton problem setup in \cite{FSZ}, we face the additional difficulty of the HJB equation being ill-posed. In addition, no general estimates of the remainder were given in \cite{FSZ}, so that our approach provides new insights in the Merton problem setting as well.          

\medskip

The following subsection describes our framework.

\subsection{Setting}\label{sub_setting}

We consider $n$ tradeable securities whose prices follow the stochastic differential equations
\begin{equation}\label{market}
\mathrm{d}S_i(t)=S_i(t)\,\mu_i\big(Y^{\delta}(t),Y^{\epsilon}(t)\big)\,\mathrm{d}t 
+ S_i(t)\,\sigma_i\big(Y^{\delta}(t),Y^{\epsilon}(t)\big)^T\,\mathrm{d}W(t),\quad i=1,\,2,\,\ldots,\,n
\end{equation}
and where $Y^{\delta}$, $Y^{\epsilon}$ are two observable real-valued stochastic factors. The factors are modelled by one-dimensional diffusion processes
\begin{align}
 \mathrm{d}Y^{\delta}(t)&=\delta\,b(Y^{\delta}(t))\,\mathrm{d}t+\sqrt{\delta}\,\kappa(Y^{\delta}(t))\,\mathrm{d}B_1(t), \label{Yd} \\
 \mathrm{d}Y^{\epsilon}(t)&=\frac{1}{\epsilon}\,\gamma(Y^{\epsilon}(t))\,\mathrm{d}t
+\frac{1}{\sqrt{\epsilon}}\,\alpha(Y^{\epsilon}(t))\,\mathrm{d}B_2(t). \label{Yeps}
\end{align}
We think of $\delta$, $\epsilon$ as being small positive numbers, so that $Y^{\delta}$ should be thought of as a slow factor (e.g. a macroeconomic indicator) and $Y^{\epsilon}$ as a fast factor (e.g. a fast mean-reverting stochastic volatility). Hereby, the noise $\mathcal{W}=(W,B_1,B_2)$ is jointly a $(d+2)$-dimensional Brownian motion, $W$ is a $d$-dimensional standard Brownian motion, $B_1$, $B_2$ are one-dimensional standard Brownian motions, and the covariance structure is given by 
\begin{align}
 \rho^{\mathbf{s}}_j\,t&:=\langle W_j,B_1\rangle(t),\quad j=1,\,2,\,\ldots,\,d, \\
\rho^{\mathbf{f}}_j\,t&:=\langle W_j,B_2\rangle(t),\quad j=1,\,2,\,\ldots,\,d, \\
 \rho^{\mathbf{s,f}}\,t&:=\langle B_1,B_2\rangle(t).
\end{align} 
Since we allow for non-perfect correlation between the asset price processes and the stochastic factors, the 
market is in general \textit{incomplete}.

\medskip

In our setting the forward performance SPDE \eqref{SPDE0} reads
\begin{equation}\label{SPDE}
\begin{split}
\mathrm{d}U(t,x)=\frac{1}{2}\,\frac{\big\|U_x(t,x)\,\lambda\big(Y^{\delta}(t),Y^{\epsilon}(t)\big)
+\sigma\big(Y^{\delta}(t),Y^{\epsilon}(t)\big)\,\sigma\big(Y^{\delta}(t),Y^{\epsilon}(t)\big)^{-1}
a^W_x(t,x)\big\|^2}{U_{xx}(t,x)}\,\mathrm{d}t \\
+\,a(t,x)^T\,\mathrm{d}\widehat{\mathcal{W}}(t). 
\end{split}
\end{equation}
Here $\sigma\big(Y^{\delta}(t),Y^{\epsilon}(t)\big)=\big(\sigma_1\big(Y^{\delta}(t),Y^{\epsilon}(t)\big),\ldots,\sigma_n\big(Y^{\delta}(t),Y^{\epsilon}(t)\big)\big)$ is the volatility matrix of the stock price processes \eqref{market}, and 
\[ \lambda\big(Y^{\delta}(t),Y^{\epsilon}(t)\big)=\big(\sigma\big(Y^{\delta}(t),Y^{\epsilon}(t)\big)^T)^{-1}\mu\big(Y^{\delta}(t),Y^{\epsilon}(t)\big)\]
is the market price of risk. The superscripts $T$ and $-1$ denote transpose and Moore-Penrose pseudoinverse as before; and $\widehat{\mathcal{W}}$ is the standard Brownian motion obtained from the Brownian motion $\mathcal{W}$ by left-multiplication with a suitable constant matrix.

\medskip

In this paper we focus on solutions of \eqref{SPDE} of factor form, namely on processes represented as 
\begin{equation}\label{factor_form}
U(t,x)=V\big(t,x,Y^{\delta}(t),Y^{\epsilon}(t)\big),
\end{equation}
for some deterministic function $V=V(t,x,y_1,y_2)$. 
As discussed in the introduction of \cite{NT}, such solutions are particularly natural from the economic point of view. Indeed, thinking of the forward performance process $U$ as encoding the preferences of the investor on a set of trading strategies given the state of the world she observes and assuming that there are only \textit{finitely} many quantities the investor keeps track of, it is natural to assume that the state enters her preferences through the corresponding finite number of factor processes. 

\medskip

Assuming $V\in C^{1,2,2,2}$, applying It\^o's formula to $V\big(t,x,Y^{\delta}(t),Y^{\epsilon}(t)\big)$, equating first the resulting martingale part with the martingale part on the right-hand side of \eqref{SPDE}, and then the two bounded variation parts, one concludes that the function $V(t,x,y_1,y_2)$ is a classical solution of the HJB equation 
\begin{equation}\label{HJB1}
V_t+\mathcal{A}^{\delta,\epsilon}_y\,V
-\frac{1}{2}\,\frac{\big\|V_x\,\lambda+\sigma\,\sigma^{-1}\big(V_{xy_1}\,\sqrt{\delta}\,\kappa\,\rho^{\mathbf{s}}
+V_{xy_2}\,\frac{1}{\sqrt{\epsilon}}\,\alpha\,\rho^{\mathbf{f}}\big)\big\|^2}{V_{xx}}=0\,.
\end{equation} 
Here $\mathcal{A}^{\delta,\epsilon}_y$ is the generator of the diffusion process $\big(Y^{\delta},Y^{\epsilon}\big)$. 

\medskip

We also note that the initial condition $U(0,\cdot)$ for the SPDE \eqref{SPDE} translates into an initial condition $V(0,x,y_1,y_2)$ for the HJB equation \eqref{HJB1}, so that the latter is posed in the ``wrong'' time direction and, in particular, one does not expect solutions to exist for all initial conditions or to depend continuously on them. In general, this ill-posedness is the main mathematical difficulty in dealing with forward performance processes.  

\medskip

The main results of the paper (Proposition \ref{multi_prop}, Theorem \ref{multi_conv_thm}) identify explicitly the leading order and first order correction terms of the solution $V$ of \eqref{HJB1} in the limit regime $\delta\downarrow0$, $\epsilon\downarrow0$. This allows to identify the leading order and first order correction terms of the corresponding pair $(U,\pi^*)$ explicitly as well (see Proposition \ref{opt_prop}). All our results are obtained under the following two assumptions.

\begin{asmp}\label{main_asmp}
\begin{enumerate}[(i)]
\item The range of left-multiplication by the matrix $\sigma$ is all of $\rr^d$, so that $\sigma\,\sigma^{-1}$ is the $d\times d$ identity matrix. In particular, this implies $n\geq d$ (and so the incompleteness in the model stems from the imperfectly correlated factors). Further, $\lambda$ is smooth in $(y_1,y_2)$.
\item The initial condition does not depend on the factors, and we will write it as $V(0,x)$. Further $V(0,x)$ is increasing and strictly concave in $x$.
\item The function $(V_x)^{(-1)}(0,e^{-x})$ 
admits the representation 
\begin{equation}
(V_x)^{(-1)}(0,e^{-x})= \int_\rr \frac{e^{zx}-1}{z}\,\nu_0(\mathrm{d}z) +C_0, 
\end{equation}
for some non-negative finite Borel measure $\nu_0$ on $\rr$ and constant $C_0\in\rr$, where $(V_x)^{(-1)}(0,\cdot)$ is the inverse of the function $V_x(0,\cdot)$. 
\end{enumerate}
\end{asmp}
The latter condition is related to the ill-posedness of the initial value problem for the HJB equation \eqref{HJB1}, and will turn out to be necessary for the leading order term of $V$ to be well-defined. A class of possible initial conditions is given by $V(0,x)=c_1\,x^{c_2}$, $c_1>0$, $c_2\in(0,1)$.

\begin{asmp}\label{inv_asmp}
The process $Y^{\epsilon}$ is positive recurrent with a unique invariant distribution $\mu$. Clearly, the latter does not depend on the value of $\epsilon$ (since a change in $\epsilon$ corresponds to a multiplication of the generator of $Y^{\epsilon}$ by a constant).
\end{asmp}



\subsection{Outline}

To ensure that the main ideas are not obscured by cumbersome notation we first consider the cases where only the slow factor $Y^{\delta}$ is present (``slow factor case'', Section \ref{sec_slow}) or only the fast factor $Y^{\epsilon}$ is present (``fast factor case'', Section \ref{sec_fast}). 

In the slow factor case we provide explicit formulas for the leading order and first order correction terms of $V$ in Propositions \ref{slow_lead} and \ref{CorrTerm}, and justify the approximation of $V$ by such in Theorem \ref{conv_prop}. The corresponding results in the fast factor case can be found in Propositions \ref{ffV0prop} and \ref{fast_corr_prop}, and Theorem \ref{fast_conv_thm}. In Section \ref{sec_multi} we consider the general case and give explicit formulas for the leading order and first order correction terms of $V$ in Proposition \ref{multi_prop}. The corresponding remainder estimate can be found in Theorem \ref{multi_conv_thm}. Finally, in Section \ref{sec_port} we give explicit formulas for the portfolios associated with our approximation (Definition \ref{opt_def}) and explain in which sense they are approximately optimal (Proposition \ref{opt_prop}).

\section{Forward investment problem with a slow factor} \label{sec_slow}

The first situation we consider is the slow factor case, that is when $\mu_i$ and $\sigma_i$ in \eqref{market} depend only on $Y^\delta$, and so $V(t,x,y_1,y_2)$ in \eqref{factor_form} does not depend on $y_2$. 
Moreover, to simplify the notation we write $\rho$ for $\rho^{\mathbf{s}}$ and $y$ for $y_1$ throughout the present section. In view of Assumption \ref{main_asmp}, and with these notations, the HJB equation \eqref{HJB1} becomes 
\begin{equation}\label{HJB2}
V_t+\frac{1}{2}\,\delta\,\kappa(y)^2\,V_{yy}+\delta\,b(y)\,V_y-\frac{1}{2}\,\frac{\|V_x\,\lambda(y)+V_{xy}\,\sqrt{\delta}\,\kappa(y)\,\rho\|^2}{V_{xx}}=0. 
\end{equation}
Here we aim to find an expansion of $V$ of the form
\begin{equation}\label{expansion}
V=V^{(0)}+\sqrt{\delta}\,V^{(1)}+O(\delta)
\end{equation}
in the limit regime $\delta\downarrow0$. To this end, we will first derive expressions for $V^{(0)}$ and $V^{(1)}$ informally and then justify the resulting expansion in Theorem \ref{conv_prop} below.

\subsection{Asymptotic Analysis}
To obtain the leading order term $V^{(0)}$ we set $\delta=0$ in \eqref{HJB2}:
\begin{equation}\label{V0eq}
V^{(0)}_t-\frac{1}{2}\,\frac{\|\lambda(y)\|^2\,\big(V^{(0)}_x\big)^2}{V^{(0)}_{xx}}=0.
\end{equation}
In addition, we endow the latter equation with the initial condition $V^{(0)}(0,x,y)=V(0,x)$. The resulting problem corresponds to taking the volatility coefficient in the forward performance SPDE \eqref{SPDE} to be zero. This is precisely the case of time-monotone forward performance processes studied in \cite{MZ4}. The formula for the solution $V^{(0)}$ of \eqref{V0eq} can be therefore recovered directly from \cite[Theorems 4 and 8]{MZ4}.

\begin{proposition}[Leading order term, slow factor]\label{slow_lead}
The solution $V^{(0)}$ of the HJB equation \eqref{V0eq} with the initial condition $V^{(0)}(0,x,y)=V(0,x)$ admits the following representation in terms of $\nu_0$ and $C_0$ in Assumption \ref{main_asmp}(iii):
\begin{equation}
V^{(0)}(t,x,y)=u\Big(\|\lambda(y)\|^2\,t,x\Big) \label{V0formula}
\end{equation}
where $u$ is given by
\begin{align}
u(t,x)&=-\frac{1}{2}\,\int_0^t e^{-h^{(-1)}(s,x)+\frac{s}{2}}\,h_x\big(s,h^{(-1)}(s,x)\big)\,\mathrm{d}s+V(0,x), \label{udef} \\
 h(t,x)&=\int_\rr \frac{e^{zx-\frac{1}{2}z^2t}-1}{z}\,\nu_0(\mathrm{d}z) +C_0,
\end{align}
and $h^{(-1)}$ denotes the inverse of $h$ in the variable $x$. 
\end{proposition} 
This follows from a transformation of \eqref{V0eq} to the ill-posed heat equation and Widder's Representation Theorem of positive solutions to this equation (\cite[Theorem 8.1]{Wi}). Both this transformation and Widder's Theorem will be used to construct higher order terms of the expansion. 
The interpretation of \eqref{V0formula} is that, at principal order, the forward performance measure is the complete market solution, but with the Sharpe ratio frozen to $\lambda(y)$.

\smallskip

Next, we turn to the correction term $V^{(1)}$ in \eqref{expansion}. To obtain an equation for $V^{(1)}$ we plug $V^{(0)}+\sqrt{\delta}\,V^{(1)}$ into \eqref{HJB2} and collect the terms of order $\sqrt{\delta}$. To this end, we note the expansions
\begin{eqnarray*}
&&\begin{split}
&\|\big(V^{(0)}+\sqrt{\delta}\,V^{(1)}\big)_x\,\lambda+\big(V^{(0)}+\sqrt{\delta}\,V^{(1)}\big)_{xy}\,\sqrt{\delta}\,\kappa\,\rho\|^2 \\
&=\big(V^{(0)}_x\big)^2\|\lambda\|^2
+\sqrt{\delta}\Big(2V^{(0)}_x\,V^{(1)}_x\,\|\lambda\|^2+2V^{(0)}_x\,V^{(0)}_{xy}\,\kappa\,\lambda^T\rho\Big)+O(\delta)\,, 
\end{split} \\
&&\frac{1}{V^{(0)}_{xx}+\sqrt{\delta}\,V^{(1)}_{xx}}=\frac{1}{V^{(0)}_{xx}}
-\sqrt{\delta}\,\frac{V^{(1)}_{xx}}{\big(V^{(0)}_{xx}\big)^2}+O(\delta)\,.
\end{eqnarray*}
The resulting equation for $V^{(1)}$ reads
\begin{equation}\label{V1slow}
V^{(1)}_t+\frac{1}{2}\,\frac{\|\lambda(y)\|^2\,\big(V^{(0)}_x\big)^2}{\big(V^{(0)}_{xx}\big)^2}\,V^{(1)}_{xx}
-\frac{\|\lambda(y)\|^2\,V^{(0)}_x}{V^{(0)}_{xx}}\,V^{(1)}_x
=\kappa(y)\,\lambda(y)^T\rho\,\frac{V^{(0)}_x\,V^{(0)}_{xy}}{V^{(0)}_{xx}}
\end{equation}
and it is endowed with the initial condition $V^{(1)}(0,x,y)=0$, since the zeroth order term $V^{(0)}$ has already satisfied the initial condition for $V$.

\begin{proposition}[Correction term, slow factor]\label{CorrTerm}
The solution $V^{(1)}$ of the PDE \eqref{V1slow} endowed with the initial condition $V^{(1)}(0,x,y)=0$ is given by
\begin{equation}
V^{(1)}(t,x,y)=\frac{t}{2}\,\kappa(y)\,\lambda(y)^T\rho\,\frac{V^{(0)}_x(t,x,y)\,V^{(0)}_{xy}(t,x,y)}
{V^{(0)}_{xx}(t,x,y)}. \label{V1slowformula}
\end{equation}
\end{proposition}

\smallskip

\begin{proof}
We start by introducing the change of variables
\begin{equation}\label{CoVslow}
(t,\xi,y):=\Big(t,-\log V^{(0)}_x-\frac{\|\lambda(y)\|^2}{2}\,t,y\Big),
\end{equation}
and set $w^{(0)}(t,\xi,y)=V^{(0)}(t,x,y)$, $w^{(1)}(t,\xi,y)=V^{(1)}(t,x,y)$. The latter functions are well-defined, since $V^{(0)}$ is strictly increasing and strictly concave in $x$, so that $\xi$ is a strictly increasing function of $x$. 

\medskip

By viewing the equation \eqref{V0eq} as the ``linear'' equation
\begin{equation}\label{V0slowlin}
V^{(0)}_t+\frac{\|\lambda(y)\|^2}{2}\,\left(\frac{V^{(0)}_x}{V^{(0)}_{xx}}\right)^2\,V^{(0)}_{xx}
-\|\lambda(y)\|^2\,\left(\frac{V^{(0)}_x}{V^{(0)}_{xx}}\right)\,V^{(0)}_x=0
\end{equation}
with coefficients depending on $V^{(0)}$, computing the derivatives
\begin{align}
 V^{(0)}_t &= w^{(0)}_t+\frac{\|\lambda(y)\|^2}{2}\,w^{(0)}_\xi\,
\bigg(-\bigg(\frac{V^{(0)}_x}{V^{(0)}_{xx}}\bigg)_x-2\bigg), \label{wslowt}\\
 V^{(0)}_x &= w^{(0)}_\xi\,\bigg(-\frac{V^{(0)}_{xx}}{V^{(0)}_x}\bigg), \label{wslowx} \\
 V^{(0)}_{xx} &= w^{(0)}_{\xi\xi}\,\bigg(\frac{V^{(0)}_{xx}}{V^{(0)}_x}\bigg)^2 
+w^{(0)}_\xi\,\,\bigg(\frac{V^{(0)}_{xx}}{V^{(0)}_x}\bigg)^2\,
\bigg(\frac{V^{(0)}_x}{V^{(0)}_{xx}}\bigg)_x, \label{wslowxx}
\end{align}
and plugging \eqref{wslowt}, \eqref{wslowx} and \eqref{wslowxx} into \eqref{V0slowlin}, we obtain
\begin{equation}
w^{(0)}_t+\frac{\|\lambda(y)\|^2}{2}\,\left(\frac{V^{(0)}_x}{V^{(0)}_{xx}}\right)^2
\bigg(w^{(0)}_{\xi\xi}\,\bigg(\frac{V^{(0)}_{xx}}{V^{(0)}_x}\bigg)^2
+w^{(0)}_\xi\,\bigg(-\frac{V^{(0)}_x}{V^{(0)}_{xx}}\bigg)_x\bigg)
-\|\lambda(y)\|^2\left(\frac{V^{(0)}_x}{V^{(0)}_{xx}}\right)w^{(0)}_\xi\,\bigg(-\frac{V^{(0)}_{xx}}{V^{(0)}_x}\bigg)=0.
\end{equation}
This readily simplifies to
\begin{equation}\label{w0heateq}
w^{(0)}_t+\frac{\|\lambda(y)\|^2}{2}\,w^{(0)}_{\xi\xi}=0
\end{equation}
in the coordinates $(t,\xi,y)$. 

A similar computation for $w^{(1)}$ shows that \eqref{V1slow} transforms into
\begin{equation}\label{slowheatw1}
w^{(1)}_t+\frac{\|\lambda(y)\|^2}{2}\,w^{(1)}_{\xi\xi}
=\kappa\,\lambda(y)^T\,\rho\,t\,\lambda(y)^T\,\lambda'(y)\,w^{(0)}_{\xi\xi\xi}
\end{equation}
in the new variables. Hereby, to obtain the right-hand side of \eqref{slowheatw1} we have relied on the following two considerations: 
\smallskip
\begin{enumerate}[(a)] 
\item Note from \eqref{udef} that $u$ is differentiable in $t$ and that $u_t$ is differentiable in $x$, which, with the assumed smoothness of $\lambda$, implies the differentiability of $w^{(0)}$ in $y$. 
Then, by differentiating \eqref{w0heateq} in $y$ and rearranging one obtains 
\begin{equation}
\big(w^{(0)}_y\big)_t+\frac{\|\lambda(y)\|^2}{2}\,\big(w^{(0)}_y\big)_{\xi\xi}=-\lambda(y)^T\,\lambda'(y)\,w^{(0)}_{\xi\xi}.
\end{equation}
Moreover, the latter equation, endowed with the initial condition $w^{(0)}_y(0,\xi,y)=0$, has the unique solution $w^{(0)}_y=-t\,\lambda(y)^T\,\lambda'(y)\,w^{(0)}_{\xi\xi}$ (note that the uniqueness of the solution is a consequence of Widder's Theorem). In the original coordinates this solution reads
\begin{equation}\label{V0yeq}
V^{(0)}_y=-t\,\lambda(y)^T\,\lambda'(y)\,\bigg(\frac{\big(V^{(0)}_x\big)^2}{V^{(0)}_{xx}}\bigg)_x\;
\frac{V^{(0)}_x}{V^{(0)}_{xx}}\,.
\end{equation} 
\item In addition, by using \eqref{wslowx}, \eqref{wslowxx} one finds
\begin{equation}
w^{(0)}_{\xi\xi\xi}=-\frac{V^{(0)}_x}{V^{(0)}_{xx}}\,\bigg(\bigg(\frac{\big(V^{(0)}_x\big)^2}{V^{(0)}_{xx}}\bigg)_x\;
\frac{V^{(0)}_x}{V^{(0)}_{xx}}\bigg)_x\,.
\end{equation}
\end{enumerate}

\noindent A combination of (a) and (b) gives the right-hand side of \eqref{slowheatw1}. At this point, one can check that the unique solution of \eqref{slowheatw1}, endowed with the initial condition $w^{(1)}(0,\xi,y)=0$ is given by 
\begin{equation}
w^{(1)}=\frac{t^2}{2}\,\kappa(y)\,\lambda(y)^T\,\rho\,\lambda(y)^T\,\lambda'(y)\,w^{(0)}_{\xi\xi\xi}.
\end{equation}
Hereby, the uniqueness part of the statement follows by applying again the Widder's Theorem. 
To obtain the proposition, it remains to change the coordinates back to $(t,x,y)$ and use that
\begin{equation}
V^{(0)}_{xy}=-t\,\lambda(y)^T\,\lambda'(y)\,\bigg(\bigg(\frac{\big(V^{(0)}_x\big)^2}{V^{(0)}_{xx}}\bigg)_x\;
\frac{V^{(0)}_x}{V^{(0)}_{xx}}\bigg)_x,
\end{equation} 
which can be obtained from \eqref{V0yeq} by a differentiation in $x$.
\end{proof}
\begin{rmk}\label{overlaps}
This result can be considered as the forward performance analog of \cite[Proposition 3.3]{FSZ} for the (backwards in time) Merton problem, but here the transformation \eqref{CoVslow} is crucial to reduce to the ill-posed heat equation for which Widder's Theorem can be applied. In the traditional utility maximization problem, the correction term can be computed directly using commutation of certain operators and Black's (fast diffusion) equation for the risk tolerance function, which is the approach used in \cite{FSZ}. Moreover in \cite{FSZ}, only the case of a single stock is considered, and here the analysis is done with multiple assets to obtain the correct coefficients in the approximations, which could not be obtained from just the one asset case with simple parameter replacements. 
\end{rmk}

We see from \eqref{V1slowformula} that $V^{(1)}$ depends on the slow factor through the parameters $\kappa$ and $\lambda$ frozen to the values $\kappa(y)$ and $\lambda(y)$, as well as $\lambda'(y)$. It is easily computed in terms of derivatives of the complete market forward performance measure $V^{(0)}$. We now give an additional representation of the correction term $V^{(1)}$ which has a natural interpretation in terms of the original forward performance problem. 

\begin{proposition}[Natural parametrization of correction term, slow factor]\label{slownpprop}
The solution $V^{(1)}$ of the PDE \eqref{V1slow} endowed with the initial condition $V^{(1)}(0,x,y)=0$, written as $w^{(1)}(t,\xi,y)$ in the coordinates $(t,\xi,y)$ defined in \eqref{CoVslow},
admits the representation
\begin{equation}\label{slownpeq}
w^{(1)}(t,\xi,y)=\int_0^t w^{(1),s}(t,\xi,y)\,\mathrm{d}s
\end{equation}
where each $w^{(1),s}$ is the solution of the initial value problem
\eq
w^{(1),s}_t+\frac{\|\lambda(y)\|^2}{2}\,w^{(1),s}_{\xi\xi}=0,\quad t\ge s \label{w1sslow} 
\en
with initial condition
\eq
w^{(1),s}(s,\xi,y)=s\,\kappa(y)\,\lambda(y)^T\,\rho\,\|\lambda(y)\|^2\,w^{(0)}_{\xi\xi}(s,\xi,y). 
\label{w1sslowic}
\en
In particular, each $w^{(1),s}$ can be represented as
\begin{equation}\label{nusslow}
w^{(1),s}(t,\xi,y)=\int_\rr e^{z\xi-z^2(t-s)}\,\nu^{s,y}(\mathrm{d}z)
\end{equation}
with $\nu^{s,y}$ being a signed finite Borel measure on $\rr$. 
\smallskip

In the original coordinates, the same representation reads
\begin{equation}\label{V1slowrepr}
V^{(1)}(t,x,y)=\int_0^t V^{(1),s}(t,x,y)\,\mathrm{d}s
\end{equation}
where each $V^{(1),s}$ is the solution of the initial value problem
\eq
V^{(1),s}_t+\frac{\|\lambda(y)\|^2}{2}\,\frac{\big(V^{(0)}_x\big)^2}{\big(V^{(0)}_{xx}\big)^2}\,V^{(1),s}_{xx}
-\|\lambda(y)\|^2\,\frac{V^{(0)}_x}{V^{(0)}_{xx}}\,V^{(1),s}_x=0\,,\quad t\ge s
\en
with initial condition
\eq
V^{(1),s}(s,x,y)=\kappa(y)\,\lambda(y)^T\rho\,\frac{V^{(0)}_x(s,x,y)\,V^{(0)}_{xy}(s,x,y)}
{V^{(0)}_{xx}(s,x,y)}\,.
\en
\end{proposition}

\smallskip

\begin{rmk}\label{rmk_w1s}
We remark that each of the processes $V^{(1),s}$ (or, equivalently, $w^{(1),s}$) can be viewed as an ``auxiliary'' forward performance process. These should be interpreted as the first order corrections that the investor makes at any given time $s$ in reaction to the market conditions she observes. Furthermore, when making a projection of her future preferences from a time $t$ onwards, the investor corrects her leading order forward performance criterion $V^{(0)}$ (or, equivalently, $w^{(0)}$) by aggregating all her previous first order corrections $V^{(1),s}$, $s\in[0,t]$ (or, equivalently, $w^{(1),s}$). We refer to $V^{(1),s}$ as auxiliary to stress that it is neither a complete market forward performance measure, nor a solution of the full incomplete market problem.
\end{rmk}

\smallskip

\noindent\textit{Proof of Proposition \ref{slownpprop}.} We first recall that $w^{(0)}$ is a classical solution of the forward heat equation \eqref{w0heateq}. Hence, the same is true for $w^{(0)}_{\xi\xi}$. At this point, an application of Widder's Theorem shows that the initial value problem \eqref{w1sslow}, \eqref{w1sslowic} has a solution that exists for all $t\ge s$. In other words, each function $w^{(1),s}$, $s\ge0$ is well-defined. Since the forward heat equation with source \eqref{slowheatw1} has a unique classical solution starting from the zero initial condition by Widder's Theorem, the representation \eqref{slownpeq} will follow once we establish that the right-hand side of \eqref{slownpeq} is a classical solution of \eqref{slowheatw1}. This is the result of the following computation:
\begin{eqnarray*}
\partial_t\left(\int_0^t w^{(1),s}(t,\xi,y)\,\mathrm{d}s\right)
&=&w^{(1),t}(t,\xi,y)+\int_0^t w^{(1),s}_t(t,\xi,y)\,\mathrm{d}s \\
&=&t\,\kappa(y)\,\lambda(y)^T\,\rho\,\|\lambda(y)\|^2\,w^{(0)}_{\xi\xi}(t,\xi,y)
-\int_0^t \frac{\|\lambda(y)\|^2}{2}\,w^{(1),s}_{\xi\xi}(t,\xi,y)\,\mathrm{d}s \\
&=&t\,\kappa(y)\,\lambda(y)^T\,\rho\,\|\lambda(y)\|^2\,w^{(0)}_{\xi\xi}(t,\xi,y)
-\frac{\|\lambda(y)\|^2}{2}\,\partial_{\xi\xi}\left(\int_0^t w^{(1),s}(t,\xi,y)\,\mathrm{d}s\right).
\end{eqnarray*}
Finally, \eqref{nusslow} follows applying again Widder's Theorem \cite[Theorem 8.1]{Wi}, and the representation \eqref{V1slowrepr} is the result of writing \eqref{slownpeq} in the original coordinates $(t,x,y)$. \ep

\bigskip

The next theorem shows that, under appropriate assumptions, the error in the approximation of the true value function $V$ by $V^{(0)}+\sqrt{\delta}\,V^{(1)}$ is indeed of order $\delta$, as one would expect. To this end, we define the non-linear functional
\begin{equation}
\eta^{\delta}:=\frac{\|\lambda\|^2}{2\delta}\,V_{xx}\,\bigg(\frac{V_x}{V_{xx}}-\frac{V^{(0)}_x}{V^{(0)}_{xx}}\bigg)^2
+\frac{\kappa\,\lambda^T\,\rho}{\sqrt{\delta}}\,
\bigg(\frac{V_x\,V_{xy}}{V_{xx}}-\frac{V^{(0)}_x\,V^{(0)}_{xy}}{V^{(0)}_{xx}}\bigg) 
 +\frac{1}{2}\,\kappa^2\,\|\rho\|^2\,\frac{V_{xy}^2}{V_{xx}}-\frac{1}{2}\,\kappa^2\,V_{yy}-b\,V_y\,,
\label{whatiseta1}
\end{equation}
recall the change of coordinates $(t,x,y)\mapsto(t,\xi,y)$ of \eqref{CoVslow}, and set 
\begin{equation}
\tilde{\eta}^\delta(t,\xi,y)=\eta^{\delta}(t,x,y). \label{whatiseta}
\end{equation}
The bound on the approximation error can be then stated as follows.

\begin{theorem}[Remainder estimate, slow factor]\label{conv_prop}
Suppose that there exist $\delta_0>0$ and $T\le\infty$ such that for all $\delta\in(0,\delta_0)$ the HJB equation \eqref{HJB2} has a solution $V\in C^{1,2,2}\big([0,T)\times(0,\infty)\times\rr\big)$ which is increasing and strictly concave in the second argument. Then, 
\begin{enumerate}[(i)]
\item the quantity
\begin{equation} \label{etacontr1}
\int_\rr e^{-\frac{z^2}{2\,\|\lambda\|^2\,t}}\,
\sum_{k=0}^\infty \frac{(-1)^k\,z^{2k}}{(2k)!\,2^k\,\|\lambda\|^{2k}\,t^k}
\bigg(\frac{\mathrm{d}}{\mathrm{d}\xi}\bigg)^{2k}\,\int_0^t \int_\rr \tilde{\eta}^{\delta}(s,\xi-\chi,y)\,
s^{-1/2}\,e^{-\frac{\chi^2}{2\,\|\lambda\|^2\,s}}\,\mathrm{d}\chi\,\mathrm{d}s\,\mathrm{d}z,
\end{equation}
with $\tilde{\eta}^{\delta}$ as in \eqref{whatiseta}, is well-defined and finite for all $\delta\in(0,\delta_0)$, and
\item for every $(t,x,y)\in[0,T)\times(0,\infty)\times\rr$ for which the limit superior
\begin{equation} \label{etacontr2}
 \quad \underset{\delta\downarrow0}{\overline{\lim}} \;\, \bigg|\int_\rr e^{-\frac{z^2}{2\,\|\lambda\|^2\,t}}\,
\sum_{k=0}^\infty \frac{(-1)^k\,z^{2k}}{(2k)!\,2^k\,\|\lambda\|^{2k}\,t^k}
\bigg(\frac{\mathrm{d}}{\mathrm{d}\xi}\bigg)^{2k}\,\int_0^t \int_\rr \tilde{\eta}^{\delta}(s,\xi-\chi,y)\,
s^{-1/2}\,e^{-\frac{\chi^2}{2\,\|\lambda\|^2\,s}}\,\mathrm{d}\chi\,\mathrm{d}s\,\mathrm{d}z \bigg|
\end{equation}
is finite, the error bound
\begin{equation}\label{main_err_est}
\underset{\delta\downarrow0}{\overline{\lim}}\;\delta^{-1}\,\Big|V(t,x,y)-V^{(0)}(t,x,y)-\sqrt{\delta}\,V^{(1)}(t,x,y)\Big|<\infty 
\end{equation}
applies. If the limit superior in \eqref{etacontr2} is bounded above uniformly on a subset of $[0,T)\times(0,\infty)\times\rr$, then the limit superior in \eqref{main_err_est} is bounded above uniformly on the same subset of $[0,T)\times(0,\infty)\times\rr$.  
\end{enumerate}
\end{theorem}

\begin{rmk}\label{conv_rmk}
The meaning of condition \eqref{etacontr2} can be understood as follows. As explained in the proof of Theorem \ref{conv_prop} below, the nonlinearity $\eta^\delta$ arising in the expansion of the HJB equation \eqref{HJB2} is fed into an initial value problem for a backward heat equation through a source term in the new coordinates $(t,\xi,y)$. The latter problem is severely ill-posed, with its solution operator rapidly magnifying the inverse Laplace modes (i.e. the analogues of Fourier modes for the inverse Laplace transform operator) of the source term. Therefore, in order to control the solutions uniformly for all small positive $\delta$, one needs an a priori estimate on the inverse Laplace transform of the source term that is uniform for all small positive $\delta$. The latter is precisely the content of condition \eqref{etacontr2}. In fact, the proof of Theorem \ref{conv_prop} reveals that condition \eqref{etacontr2} is sharp in the sense of \eqref{etacontr2} being equivalent to \eqref{main_err_est}.   
\end{rmk}

\smallskip

\noindent\textit{Proof of Theorem \ref{conv_prop}.} We start by expressing the HJB equation \eqref{HJB2} as
\begin{equation}\label{HJB2rewrite}
V_t-\frac{\|\lambda\|^2}{2}\,\frac{V_x^2}{V_{xx}}-\sqrt{\delta}\,\kappa\,\lambda^T\,\rho\,\frac{V_x\,V_{xy}}{V_{xx}}
=\frac{\delta}{2}\,\kappa^2\,\|\rho\|^2\,\frac{V_{xy}^2}{V_{xx}}-\frac{\delta}{2}\,\kappa^2\,V_{yy}-\delta\,b\,V_y\,.
\end{equation}
Next, we write $V=V^{(0)}+\sqrt{\delta}\,V^{(1)}+\sqrt{\delta}\,Q$, insert the latter expression into the left-hand side of \eqref{HJB2rewrite}, and expand in $\sqrt{\delta}$ using the elementary identity
\begin{equation}
\frac{1}{a+\sqrt{\delta}\,b}=\frac{1}{a}-\sqrt{\delta}\,\frac{b}{a^2+\sqrt{\delta}\,ab},\qquad a<0,\;\; 
b<-\frac{a}{\sqrt{\delta}}\,.
\end{equation}
Recalling that the functions $V^{(0)}$ and $V^{(1)}$ were constructed in such a way that all terms in \eqref{HJB2rewrite} on the orders of $1$ and $\sqrt{\delta}$ cancel out, and collecting the remaining terms we obtain after a lengthy but straightforward computation that
\begin{equation}\label{S_PDE}
\sqrt{\delta}\,Q_t-\sqrt{\delta}\,\|\lambda\|^2\,\frac{V^{(0)}_x}{V^{(0)}_{xx}}\,Q_x
+\sqrt{\delta}\,\frac{\|\lambda\|^2}{2}\,\frac{\big(V^{(0)}_x\big)^2}{\big(V^{(0)}_{xx}\big)}\,Q_{xx}=\delta\,\eta^\delta.
\end{equation} 

\medskip

Next, we let $\tilde{Q}:=\sqrt{\delta}\,Q$, recall the change of coordinates 
\begin{equation*}
(t,\xi,y):=\Big(t,-\log V^{(0)}_x(t,x,y)-\frac{\|\lambda(y)\|^2}{2}\,t,y\Big)
\end{equation*}
of \eqref{CoVslow}, and define $q(t,\xi,y):=\tilde{Q}(t,x,y)$. By the same computation as in the proof of Proposition \ref{CorrTerm}, the partial differential equation \eqref{S_PDE} can be rewritten as
\begin{equation}\label{rPDE}
q_t+\frac{\|\lambda\|^2}{2}\,q_{\xi\xi}=\delta\,\tilde{\eta}^{\delta}
\end{equation} 
where $\tilde{\eta}^{\delta}(t,\xi,y)=\eta^\delta(t,x,y)$ as before. Then, Duhamel's principle for the backward equation \eqref{rPDE} implies that
\begin{equation}\label{Duhamel}
\frac{1}{\sqrt{2\pi t}\,\|\lambda\|}\,\int_\rr q(t,\xi-\chi,y)\,e^{-\frac{\chi^2}{2\,\|\lambda\|^2\,t}}
\,\mathrm{d}\chi
=\delta\,\int_0^t \frac{1}{\sqrt{2\pi s}\,\|\lambda\|}
\,\int_\rr \tilde{\eta}^{\delta}(s,\xi-\chi,y)\,e^{-\frac{\chi^2}{2\,\|\lambda\|^2\,s}}\,\mathrm{d}\chi\,\mathrm{d}s.
\end{equation} 
It follows that the right-hand side of the latter equation is in the domain of the inverse Weierstrass transform in the sense of \cite[equations (5), (6)]{Wi2}. This, in turn, implies that the quantity in \eqref{etacontr1} is well-defined and finite for all $\delta\in(0,\delta_0)$. Moreover, applying the inverse Weierstrass transform to both sides of \eqref{Duhamel}, we see that, up to a multiplicative constant, the function $\xi\mapsto \delta^{-1}\,q(t,\xi,y)$ is given by the quantity inside the absolute value in \eqref{etacontr2}. The statement of the theorem is now immediate. \ep

\subsection{Example: Forward Performance Process of Power Type} \label{new_ex}
We illustrate the results with the family of power utility forward performance processes. For a constant risk aversion coefficient $\gamma\in(0,\infty)\backslash\{1\}$, we impose the initial condition for the HJB equation \eqref{HJB1} to be
\[ V(0,x) = \gamma^\gamma\,\frac{x^{1-\gamma}}{1-\gamma}. \]
This corresponds to Examples 16 and 18 in \cite{MZ4}.

\medskip
 
We focus on an example where there is an exact solution to the problem with one factor volatility. Namely, we take
a market model
\begin{align}
\mathrm{d}S_i(t)&=S_i(t)\,\mu_i\big(Y^{\delta}(t)\big)\,\mathrm{d}t 
+ S_i(t)\,\sigma_i\big(Y^{\delta}(t)\big)^T\mathrm{d}W(t),\quad i=1,\,2,\,\ldots,\,n\\
\mathrm{d}Y^\delta(t) &= \delta(m-Y^\delta(t))\,\mathrm{d}t +\sqrt{\delta}\,\beta\,\sqrt{Y^\delta(t)}\,\mathrm{d}B_1(t),
\end{align}
where $\sigma_i(y)$ and $\mu_i(y)$ are chosen such that
\begin{equation}
\lambda(y) = \Lambda\,\sqrt{y}, \qquad \mbox{for some } \Lambda\in\rr^d,
\end{equation}
and the Brownian motions $W_j$ are correlated with the Brownian motion $B_1$ according to $\langle W_j,B_1\rangle(t)=\rho_j\,t$. We further suppose that $\|\rho\|^2\neq\frac{\gamma}{\gamma-1}$. This class of model was used, for instance, for a classical infinite horizon consumption problem in \cite{CV}.

\medskip

We insert the ansatz 
\eq
V(t,x,y)=\gamma^\gamma\,\frac{x^{1-\gamma}}{1-\gamma}\,g(t,y)
\en
into the HJB equation \eqref{HJB2} to end up with
\eq
g_t+\frac{1}{2}\,\delta\,\beta^2 y\,g_{yy}+\big(\delta(m-y)+\Gamma\,\sqrt{\delta y}\,\beta\,\Lambda^T\rho\big)\,g_y+\frac{\Gamma}{2}\,\Big(\|\Lambda\|^2 y\, g+\delta\,\beta^2y\,\|\rho\|^2\,\frac{g_y^2}{g}\Big)=0, 
\en
where $\Gamma=\frac{1-\gamma}{\gamma}$ and the initial condition is $g(0,y)=1$. Next, we make the transformation $g=\Psi^q$ with the choice $q=\frac{1}{1+\Gamma\|\rho\|^2}$, which after a short computation leads to the linear equation
\eq
\Psi_t+\frac{1}{2}\delta\beta^2y\,\Psi_{yy}+\big(\delta(m-y)+\Gamma\,\sqrt{\delta y}\,\beta\,\Lambda^T\rho\big)\,\Psi_y 
+\frac{\Gamma}{2q}\,\|\Lambda\|^2 y\,\Psi=0
\en
with initial condition $\Psi(0,y)=1$. Looking for a solution of the form $\exp(A_1(t)y+A_2(t))$ we find that $A_1$, $A_2$ need to solve the ODEs
\begin{align}
& A_1'=\frac{1}{2}\,\delta\,\beta^2 A_1^2+\big(\sqrt{\delta}\,\Gamma\,\beta\,\Lambda^T\rho-\delta\big)A_1+\frac{\Gamma}{2q}\,\|\Lambda\|^2, \\
& A_2'=\delta\,m\,A_1
\end{align}
with zero initial conditions. 
Let $a_{\pm}$ be the roots of the quadratic corresponding to the right-hand side of the Riccati equation for $A_1$:
\eq
\frac{\delta\,\beta^2}{2}\,a^2+\big(\sqrt{\delta}\,\Gamma\,\beta\,\Lambda^T\rho-\delta\big)a+\frac{\Gamma\,\|\Lambda\|^2}{2q}=0.
\en
We assume that the model parameters are such that both roots are real and $a_+>0$. This is the case when $\gamma\in(0,1)$, $\Lambda^T\rho<0$, and $\Gamma\,q\,(\Lambda^T\rho)^2>\|\Lambda\|^2$; or $\gamma>1$, $\Lambda^T\rho>0$, $q<0$, and $\Gamma\,q\,(\Lambda^T\rho)^2>\|\Lambda\|^2$; or $\gamma>1$ and $q>0$. In any of these cases, we then obtain the solution
\eq
A_1(t)=a_-\,\frac{1-e^{-\Delta t}}{1-\frac{a_-}{a_+}\,e^{-\Delta t}}\quad\text{and}\quad
A_2(t)=\delta\,m\,\bigg(a_-t-\frac{2}{\delta\,\beta^2}\,\log\bigg(\frac{1-\frac{a_-}{a_+}\,e^{-\Delta t}}{1-\frac{a_-}{a_+}}\bigg) \bigg),
\en
where $\Delta$ is the square root of the discriminant of the quadratic above. We, in turn, deduce that since $a_+>0$, the solution
\eq\label{V_example}
V(t,x,y)=\gamma^\gamma\,\frac{x^{1-\gamma}}{1-\gamma}\,\exp\big(q\,(A_1(t)\,y+A_2(t))\big)
\en
is well-defined for all $t\ge0$.

\medskip

We can now check condition \eqref{etacontr2} of Theorem \ref{conv_prop} to show the convergence of our approximation 
\eq
V^{(0)}(t,x,y)+\sqrt{\delta}\,V^{(1)}(t,x,y)
=\gamma^\gamma\,\frac{x^{1-\gamma}}{1-\gamma}\,e^{-\frac{1}{2}\,\Gamma\,\|\lambda\|^2 t}
+\frac{\sqrt{\delta}}{4}\,t^2\,y\,\beta\,\Lambda^T\rho\,\|\Lambda\|^2\,\frac{\Gamma}{\gamma}\,\gamma^\gamma\,\frac{x^{1-\gamma}}{1-\gamma}\,e^{-\frac{1}{2}\,\Gamma\,\|\lambda\|^2 t},
\en
where the expression for $V^{(0)}$ follows from Proposition \ref{slow_lead} with $\nu_0$ being a Dirac mass at $\gamma^{-1}$ and $C_0=\gamma$, and the expression for $V^{(1)}$ is a consequence of Proposition \ref{CorrTerm}. 

\medskip

In \eqref{whatiseta1}, the term in the parentheses in the first summand turns out to be zero, because both $V$ and $V^{(0)}$ are multiples of $x^{1-\gamma}$, and the term in parentheses in the second summand is of order $\sqrt{\delta}$, because $V^{(0)}$ can be obtained from $V$ by setting $\delta=0$, and $V$ is smooth in $\delta$. Therefore $\eta$ is a product of $x^{1-\gamma}$ and a function of $(t,y)$ depending smoothly on $\delta$ (which follows from the smoothness of $A_1$, $A_2$), and hence, $\tilde{\eta}^\delta$ in \eqref{whatiseta} is a product of $e^{\xi\Gamma}$ and another function of $(t,y)$ depending smoothly on $\delta$. An explicit computation of the inverse Weierstrass transform in \eqref{etacontr2} reduces the expression there to
\eq
\underset{\delta\downarrow0}{\overline{\lim}}\;\bigg|\int_0^t e^{\xi\Gamma-(t-s)\Gamma^2\frac{\|\lambda\|^2}{2}}\,F(s,y;\delta)\,\mathrm{d}s\bigg|\,,
\en 
for some function $F$ depending smoothly on $\delta$. In particular, the latter limit superior is finite and so the conditions of Theorem \ref{conv_prop} are satisfied in this example. 

\section{Forward investment problem with a fast factor} \label{sec_fast}

The second situation we consider is the fast factor case, that is when $\mu_i$ and $\sigma_i$ in \eqref{market} depend only on $Y^\epsilon$, and so $V(t,x,y_1,y_2)$ in \eqref{factor_form} does not depend on $y_1$.
To simplify the notation, we write $\rho$ for $\rho^{\mathbf{f}}$ and $y$ for $y_2$ throughout this section. With these notations and in view of Assumption \ref{main_asmp}, the HJB equation \eqref{HJB1} becomes 
\begin{equation}\label{HJB4}
V_t+\frac{\alpha(y)^2}{2\epsilon}\,V_{yy}+\frac{\gamma(y)}{\epsilon}\,V_y
-\frac{1}{2}\,\frac{\|V_x\,\lambda(y)+V_{xy}\,\frac{1}{\sqrt{\epsilon}}\,\alpha(y)\,\rho\|^2}{V_{xx}}=0.
\end{equation}
Our goal here is to find an explicit expansion of the solution $V$ to \eqref{HJB4} of the form
\begin{equation}\label{expansion2}
V=V^{(0)}+\sqrt{\epsilon}\,V^{(1)}+O(\epsilon)
\end{equation}
in the limit regime $\epsilon\downarrow0$. As in the previous section, we will first derive formulas for $V^{(0)}$ and $V^{(1)}$ informally, and then justify the resulting expansion by means of a suitable remainder estimate.

\medskip

To find $V^{(0)}$ we plug \eqref{expansion2} into \eqref{HJB4} and collect the leading order terms (namely, those on the order of $\epsilon^{-1}$) to get
\begin{equation}\label{V0fast}
\frac{\alpha(y)^2}{2}\,V^{(0)}_{yy}+\gamma(y)\,V^{(0)}_y
-\frac{1}{2}\,\alpha(y)^2\|\rho\|^2\frac{\big(V^{(0)}_{xy}\big)^2}{V^{(0)}_{xx}}=0.
\end{equation} 
Note that we can satisfy \eqref{V0fast} by choosing $V^{(0)}$ as a function of $t$ and $x$ only. As we explain below, the exact choice of $V^{(0)}$ will be pinned down by the lower order terms in the expansion of \eqref{HJB4}.

\medskip

To proceed, we plug \eqref{expansion2} into \eqref{HJB4}, and collect the terms of order $\epsilon^{-1/2}$. We obtain
\begin{equation}
-\alpha(y)\,\lambda(y)^T\,\rho\,\frac{V^{(0)}_x\,V^{(0)}_{xy}}{V^{(0)}_{xx}}
-\alpha(y)^2\,\|\rho\|^2\,\frac{V^{(0)}_{xy}\,V^{(1)}_{xy}}{V^{(0)}_{xx}}
+\frac{\alpha(y)^2\,\|\rho\|^2}{2}\,\frac{\big(V^{(0)}_{xy}\big)^2\,V^{(1)}_{xx}}{\big(V^{(0)}_{xx}\big)^2} 
+\frac{\alpha(y)^2}{2}\,V^{(1)}_{yy}+\gamma(y)\,V^{(1)}_y=0.
\end{equation}
Choosing $V^{(0)}$ to be independent of $y$ as we noted earlier, the latter equation simplifies to
\begin{equation}
\frac{\alpha(y)^2}{2}\,V^{(1)}_{yy}+\gamma(y)\,V^{(1)}_y=0.
\end{equation}
Clearly, the above equation can be satisfied by choosing $V^{(1)}$ to be a function of $(t,x)$ only. As with $V^{(0)}$, the exact choice of $V^{(1)}$ will result from considering lower order terms in the expansion of \eqref{HJB4}.  

\medskip

Next, we insert the extended expansion
\begin{equation}
V=V^{(0)}+\sqrt{\epsilon}\,V^{(1)}+\epsilon\,V^{(2)}+\epsilon^{3/2}\,V^{(3)}+O(\epsilon^2)
\end{equation}
into \eqref{HJB4} in order to find the terms of order $1$. This results in the equation
\begin{equation}\label{fastV2}
V^{(0)}_t-\frac{\|\lambda(y)\|^2}{2}\,\frac{\big(V^{(0)}_x\big)^2}{V^{(0)}_{xx}}
+\frac{\alpha(y)^2}{2}\,V^{(2)}_{yy}+\gamma(y)\,V^{(2)}_y=0.
\end{equation} 
Hereby, we have used the fact that $V^{(0)}$ and $V^{(1)}$ do not depend on $y$. The next proposition shows that there is a \textit{unique} choice of $V^{(0)}$ such that $V^{(0)}(0,x,y)=V(0,x)$ and the equation \eqref{fastV2}, viewed as an elliptic partial differential equation for $V^{(2)}$, has a solution.

\begin{proposition}[Leading order term, fast factor]\label{ffV0prop}
There is a unique function $V^{(0)}$ such that $V^{(0)}(0,x,y)=V(0,x)$ and the equation \eqref{fastV2} possesses a solution $V^{(2)}$. With 
\begin{equation}
\bar{\lambda}=\bigg(\int_\rr \|\lambda(y)\|^2\,\mu(\mathrm{d}y)\bigg)^{1/2},
\end{equation}
such a function $V^{(0)}$ admits the representation 
\begin{equation}
V^{(0)}(t,x,y)=V^{(0)}(t,x)=u\big(\bar{\lambda}^2\,t,x\big)  \label{V0fastformula}
\end{equation}
where $u$ is given by
\begin{eqnarray*}
&& u(t,x)=-\frac{1}{2}\,\int_0^t e^{-h^{(-1)}(s,x)+\frac{s}{2}}\,h_x\Big(s,h^{(-1)}(s,x)\Big)\,\mathrm{d}s
+V(0,x), \\
&& h(t,x)=\int_\rr \frac{e^{zx-\frac{1}{2}z^2t}-1}{z}\,\nu_0(\mathrm{d}z)+ C_0.
\end{eqnarray*}
Here $h^{(-1)}$ is the inverse of $h$ in the variable $x$, and $\nu_0$ and $C_0$ were introduced in Assumption \ref{main_asmp}(iii). 
\end{proposition}

\begin{proof}
We start by integrating \eqref{fastV2} with respect to the invariant distribution $\mu$ of Assumption \ref{inv_asmp}. Since $V^{(0)}$ does not depend on $y$ and 
\begin{equation}
\int_\rr \Big(\frac{\alpha(y)^2}{2}\,V^{(2)}_{yy}+\gamma(y)\,V^{(2)}_y\Big)\,\mu(\mathrm{d}y)\equiv0
\end{equation}
(due to the invariance of $\mu$), we obtain
\begin{equation}\label{fastV2ave}
V^{(0)}_t-\frac{\bar{\lambda}^2}{2}\,\frac{\big(V^{(0)}_x\big)^2}{V^{(0)}_{xx}}=0.
\end{equation}
We easily conclude using \cite[Theorems 4 and 8]{MZ4}.
\end{proof}

\smallskip

From \eqref{V0fastformula} we observe that $V^{(0)}$ gives the complete market forward performance measure with constant Sharpe ratio $\bar{\lambda}$, where the appropriate averaging has been identified by the asymptotic analysis.

\medskip

To obtain $V^{(1)}$ we will expand the HJB equation \eqref{HJB4} up to order $\epsilon^{1/2}$, which however requires further information on $V^{(2)}$. As a first step, we subtract from the equation \eqref{fastV2} its averaged version \eqref{fastV2ave} to get
\begin{equation}\label{V2ODE}
\frac{\alpha(y)^2}{2}\,V^{(2)}_{yy}+\gamma(y)\,V^{(2)}_y=\frac{\|\lambda(y)\|^2-\bar{\lambda}^2}{2}\,
\frac{\big(V^{(0)}_x\big)^2}{V^{(0)}_{xx}}.
\end{equation}
We introduce the notation
\begin{equation}\label{lambdahat}
\phi(y)=\int_0^\infty \ev^y\Big[\|\lambda(Y^{1}(s))\|^2-\bar{\lambda}^2\Big]\,\mathrm{d}s,
\end{equation}
where $Y^{1}$ denotes the fast factor process, solution of the SDE \eqref{Yeps}, but with $\epsilon=1$.
Then (see, for example, \cite[Section 3.2, p. 94]{FPSS}), the solution $V^{(2)}$ of \eqref{V2ODE} admits the stochastic representation
\begin{equation}\label{whatisV2}
V^{(2)}(t,x,y)=-\frac{1}{2}\,\frac{\big(V^{(0)}_x(t,x)\big)^2}{V^{(0)}_{xx}(t,x)}\,\phi(y)
+C(t,x),
\end{equation}
where $C(t,x)$ is a function that does not depend on $y$. 

We can now expand \eqref{HJB4} up to order $\epsilon^{1/2}$ to obtain
\begin{eqnarray*}
V^{(1)}_t+\frac{\|\lambda(y)\|^2}{2}\,\left(\frac{V^{(0)}_x}{V^{(0)}_{xx}}\right)^2\,V^{(1)}_{xx}
-\frac{V^{(0)}_x}{V^{(0)}_{xx}}\,\lambda(y)^T\,\bigg(V^{(1)}_x\,\lambda(y)
+\phi'(y)\,\bigg(\frac{\big(V^{(0)}_x\big)^2}{V^{(0)}_{xx}}\bigg)_x\,\alpha(y)\,\rho\bigg)&&\\
+\frac{\alpha(y)^2}{2}\,V^{(3)}_{yy}+\gamma(y)\,V^{(3)}_y&=&0.
\end{eqnarray*}
Averaging this equation with respect to the invariant distribution $\mu$ of Assumption \ref{inv_asmp}, we obtain further
\begin{equation}\label{V1fast}
V^{(1)}_t+\frac{\bar{\lambda}^2}{2}\,\left(\frac{V^{(0)}_x}{V^{(0)}_{xx}}\right)^2\,V^{(1)}_{xx}
-\bar{\lambda}^2\,\frac{V^{(0)}_x}{V^{(0)}_{xx}}\,V^{(1)}_x
-\frac{V^{(0)}_x}{V^{(0)}_{xx}}\,\bigg(\frac{\big(V^{(0)}_x\big)^2}{V^{(0)}_{xx}}\bigg)_x\,
\left(\int_\rr \phi'(y)\,\alpha(y)\,\lambda(y)^T\,\mu(\mathrm{d}y)\right)\rho=0,
\end{equation}
which is the desired partial differential equation for $V^{(1)}$. Since $V^{(0)}$ satisfies the initial condition for $V$, we endow \eqref{V1fast} with the initial condition $V^{(1)}(0,x,y)=0$.

\begin{proposition}[Correction term, fast factor]\label{fast_corr_prop}
The unique classical solution of the partial differential equation \eqref{V1fast} with the initial condition $V^{(1)}(0,x,y)=0$ is given by
\begin{equation}\label{V1fastsol}
V^{(1)}(t,x,y)=V^{(1)}(t,x)=t\,\frac{V^{(0)}_x}{V^{(0)}_{xx}}\,\bigg(\frac{\big(V^{(0)}_x\big)^2}{V^{(0)}_{xx}}\bigg)_x\left(\int_\rr \phi'(y)\,\alpha(y)\,\lambda(y)^T\,\mu(\mathrm{d}y)\right)\rho,
\end{equation}
with $\phi$ as in \eqref{lambdahat}.
\end{proposition}

\begin{proof}
We introduce a new space variable 
\begin{equation}
\xi:=-\log V^{(0)}_x-\frac{\bar{\lambda}^2}{2}\,t.\label{CoVfast}
\end{equation}
Since $V^{(0)}$ is strictly increasing and strictly concave in $x$ for any given $t$, $\xi$ is a strictly increasing function of $x$ and we may define 
\begin{equation}
w^{(0)}(t,\xi):=V^{(0)}(t,x)\quad\text{and}\quad w^{(1)}(t,\xi):=V^{(1)}(t,x).
\end{equation}
Following the calculations in the proof of Proposition \ref{CorrTerm} gives
\begin{equation}\label{w0eq}
w^{(0)}_t+\frac{1}{2}\,\bar{\lambda}^2\,w^{(0)}_{\xi\xi}=0.
\end{equation}
A similar computation starting from \eqref{V1fast} shows that the transformed correction term $w^{(1)}$ satisfies the forward heat equation
\begin{equation}\label{w1eq}
w^{(1)}_t+\frac{1}{2}\,\bar{\lambda}^2\,w^{(1)}_{\xi\xi}
=w^{(0)}_{\xi\xi}\left(\int_\rr \phi'(y)\,\alpha(y)\,\lambda(y)^T\,\mu(\mathrm{d}y)\right)\rho\,,
\end{equation}
with the initial condition $w^{(1)}(0,\xi)=0$. At this point, it is easy to check (using \eqref{w0eq}) that 
\begin{equation}\label{w1expl}
w^{(1)}=t\,w^{(0)}_{\xi\xi}\left(\int_\rr \phi'(y)\,\alpha(y)\,\lambda(y)^T\,\mu(\mathrm{d}y)\right)\rho
\end{equation}
satisfies \eqref{w1eq} with the desired initial condition. Changing back to the original coordinates we easily obtain \eqref{V1fastsol}.

\medskip

The uniqueness part of the proposition follows from the uniqueness of the solution of the Cauchy problem for the forward heat equation (see \cite[Theorem 8.1]{Wi}). 
\end{proof}
\begin{rmk}\label{overlaps2}
This result can be considered as the forward performance analog of \cite[Proposition 2.7]{FSZ} for the (backwards in time) Merton problem, but here the transformation \eqref{CoVfast} is crucial to reduce to the ill-posed heat equation for which Widder's Theorem can be applied. The differences highlighted in Remark \ref{overlaps} apply here too.
\end{rmk}
\medskip

Next, we give an additional representation for the correction term $V^{(1)}$ which has a natural interpretation in terms of the original portfolio optimization problem.

\begin{proposition}[Natural parametrization of correction term, fast factor]\label{fastnp}
Let $w^{(0)}$, $w^{(1)}$ be the functions $V^{(0)}$, $V^{(1)}$ from Propositions \ref{ffV0prop}, \ref{fast_corr_prop}  written in the coordinates
\begin{equation}
(t,\xi):=\bigg(t,-\log V^{(0)}_x-\frac{\bar{\lambda}^2}{2}\,t\bigg).
\end{equation} 
Then: 
\begin{enumerate}[(i)]
\item $w^{(1)}$ admits the representation 
\begin{equation}\label{w1np}
w^{(1)}(t,\xi)=\int_0^t w^{(1),s}(t,\xi)\,\mathrm{d}s
\end{equation}
where each $w^{(1),s}$ is the solution of the initial value problem
\eq
w^{(1),s}_t + \frac{\bar{\lambda}^2}{2}\,w^{(1),s}_{\xi\xi}=0,\quad t\ge s, \label{auxhe1}
\en
with initial condition
\eq
w^{(1),s}(s,\xi)=w^{(0)}_{\xi\xi}(s,\xi)\left(
\int_\rr \phi'(y)\,\alpha(y)\,\lambda(y)^T\,\mu(\mathrm{d}y)\right)\rho. \label{auxhe2}
\en
It can be therefore represented as 
\begin{equation}\label{whatisnus}
w^{(1),s}(t,\xi)=\int_\rr e^{z\xi-z^2(t-s)}\,\nu^{(s)}(\mathrm{d}z),  
\end{equation}
with $\nu^{(s)}$ being a suitable signed finite Borel measure on $\rr$. 
\smallskip
\item In the original coordinates, the same representation reads
\begin{equation}\label{V1np}
V^{(1)}(t,x)=\int_0^t V^{(1),s}(t,x)\,\mathrm{d}s
\end{equation}
where each $V^{(1),s}$ is the solution of the initial value problem
\[
V^{(1),s}_t + \frac{\bar{\lambda}^2}{2}\,\frac{\big(V^{(0)}_x\big)^2}{\big(V^{(0)}_{xx}\big)^2}\,V^{(1),s}_{xx}-\bar{\lambda}^2\,\frac{V^{(0)}_x}{V^{(0)}_{xx}}\,V^{(1),s}_x=0,\quad t\ge s,
\]
with initial condition
\[
V^{(1),s}(s,x)=\frac{V^{(0)}_x(s,x)}{V^{(0)}_{xx}(s,x)}
\,\bigg(\frac{\big(V^{(0)}_x(s,x)\big)^2}{V^{(0)}_{xx}(s,x)}\bigg)_x
\left(\int_\rr \phi'(y)\,\alpha(y)\,\lambda(y)^T\,\mu(\mathrm{d}y)\right)\rho.
\]
\end{enumerate}
\end{proposition}

\smallskip

\begin{rmk}\label{fastremark}
The quantities $V^{(1),s}$ (or, equivalently, $w^{(1),s}$) of Proposition \ref{fastnp} should be interpreted in the same way as ``auxiliary'' forward performance processes as their analogues in the slow factor case. We refer to Remark \ref{rmk_w1s} above for more details, but point out that the asymptotic analysis identifies the constant vector $\left(\int_\rr \phi'(y)\,\alpha(y)\,\lambda(y)^T\,\mu(\mathrm{d}y)\right)\rho$ as the principal correcting effect of stochastic Sharpe ratio, in the case of nonzero correlations $\rho$. 
\end{rmk}

\smallskip

\noindent\textit{Proof of Proposition \ref{fastnp}.} We first recall from the proof of Proposition \ref{fast_corr_prop} that $w^{(0)}$ is a classical solution of the forward heat equation \eqref{w0eq}. Hence, $w^{(0)}_{\xi\xi}$ is a solution of the same equation and, therefore, the solutions $w^{(1),s}$, $s\ge0$ of \eqref{auxhe1}, \eqref{auxhe2} are well-defined and given by
\begin{equation}
w^{(1),s}(t,x)=w^{(0)}_{\xi\xi}(t,\xi)\,
\left(\int_\rr \phi'(y)\,\alpha(y)\,\lambda(y)^T\,\mu(\mathrm{d}y)\right)\rho\,. 
\end{equation} 
Therefore the right-hand side of the representation \eqref{w1np} is equal to the right-hand side of \eqref{w1expl}, and, thus, \eqref{w1np} immediately follows. Moreover, the representation \eqref{whatisnus} is a direct consequence of Widder's Theorem. Finally, part (ii) of the proposition can be either established in the same way as part (i), or by changing to the original coordinates $(t,x)$ in \eqref{w1np}, \eqref{auxhe1} and \eqref{auxhe2}. \ep

\medskip

We conclude this section with the appropriate remainder estimate. Specifically, we will show that the error in the approximation of the true value function $V$ by $V^{(0)}+\sqrt{\epsilon}\,V^{(1)}$ is of the order $\epsilon$. To this end, we introduce the non-linear functional
\begin{equation}\label{eta_fast}
\begin{split}
& \eta^{\epsilon} := \frac{1}{2\epsilon^{1/2}}\,\|\lambda(y)\|^2\left(\frac{V^{(0)}_x}{V^{(0)}_{xx}}\right)^2\frac{1}{V_{xx}}
\,\big(V_{xx}-V^{(0)}_{xx}\big) \\
&+\frac{1}{2\epsilon}\,\|\lambda(y)\|^2\left(\frac{V^{(0)}_x}{V^{(0)}_{xx}}\right)^2 \frac{1}{V_{xx}}\,
\big(V_{xx}-V^{(0)}_{xx}-\epsilon^{1/2}\,V^{(1)}_{xx}\big)\,\big(V_{xx}-V^{(0)}_{xx}\big) \\
& +\lambda(y)^T\,V^{(0)}_x\,\big(\lambda(y)\,V^{(1)}_x+\alpha(y)\,\rho\,V^{(2)}_{xy}\big)
\,\bigg(\frac{1}{V_{xx}}-\frac{1}{V^{(0)}_{xx}}\bigg)\\
&-\frac{1}{\epsilon}\,\|\lambda(y)\|^2\,V^{(0)}_x\,\big(V_x-V^{(0)}_x-\epsilon^{1/2}\,V^{(1)}_x\big)
\,\bigg(\frac{1}{V_{xx}}-\frac{1}{V^{(0)}_{xx}}\bigg) \\
& +\frac{1}{\epsilon^{3/2}}\,\lambda(y)^T\,\alpha(y)\,\rho\,\frac{V^{(0)}_x}{V_{xx}}
\,\big(V_{xy}-\epsilon\,V^{(2)}_{xy}\big)
+\frac{1}{2\epsilon}\,\frac{1}{V_{xx}}\,\big\|\lambda(y)\,V_x-\lambda(y)\,V^{(0)}_x
+\epsilon^{-1/2}\,\alpha(y)\,\rho\,V_{xy}\big\|^2 \\
& +\frac{1}{\epsilon^2}\,\bigg(\bigg(V_t-\frac{1}{2}\,\frac{\|V_x\,\lambda(y)+V_{xy}\,\frac{1}{\sqrt{\epsilon}}
\,\alpha(y)\,\rho\|^2}{V_{xx}}\bigg)
-\epsilon\,\bigg(V^{(0)}_t-\frac{\|\lambda(y)\|^2}{2}\,\frac{\big(V^{(0)}_x\big)^2}{V^{(0)}_{xx}}\bigg) \\
&  -\epsilon^{3/2}\,\bigg(V^{(1)}_t+\frac{\|\lambda(y)\|^2}{2}
\left(\frac{V^{(0)}_x}{V^{(0)}_{xx}}\right)^2 V^{(1)}_{xx}
-\left(\frac{V^{(0)}_x}{V^{(0)}_{xx}}\right)\lambda(y)^T\bigg(V^{(1)}_x\,\lambda(y)
+\phi'(y)\,\bigg(\frac{\big(V^{(0)}_x\big)^2}{V^{(0)}_{xx}}\bigg)_x\,\alpha(y)\,\rho\bigg)\bigg)\bigg).
\end{split}
\end{equation}
Here $V^{(2)}$ is defined through \eqref{whatisV2}, and we note that the value of $\eta^{\epsilon}$ does not depend on the choice of the constant $C(t,x)$ in \eqref{whatisV2}. We also set $(t,\xi,y):=\big(t,-\log V^{(0)}_x-\frac{\|\lambda(y)\|^2}{2}\,t,y\big)$, and let $\tilde{\eta}^{\epsilon}(t,\xi,y)=\eta^{\epsilon}(t,x,y)$. 

\begin{theorem}[Remainder estimate, fast factor]\label{fast_conv_thm}
Suppose that there exist $\epsilon_0>0$ and $T\le\infty$ such that for all $\epsilon\in(0,\epsilon_0)$ the HJB equation \eqref{HJB4} has a solution $V\in C^{1,2,2}\big([0,T)\times(0,\infty)\times\rr\big)$ which is increasing and strictly concave in the second argument. Then, 
\begin{enumerate}[(i)]
\item the quantity
\begin{equation}
\int_\rr e^{-\frac{z^2}{2t}}\,\sum_{k=0}^\infty \frac{(-1)^k\,z^{2k}}{(2k)!\,2^k\,t^k}\,
\bigg(\frac{\mathrm{d}}{\mathrm{d}\xi}\bigg)^{2k}\,\int_0^t \int_\rr \tilde{\eta}^{\epsilon}(s,\xi-\chi,y)\,s^{-1/2}
\,e^{-\frac{\chi^2}{2s}}\,\mathrm{d}\chi\,\mathrm{d}s\,\mathrm{d}z
\end{equation}
(defined via \eqref{eta_fast} and the paragraph following it) is well-defined and finite for all $\epsilon\in(0,\epsilon_0)$, and 
\item for every $(t,x,y)\in[0,T)\times(0,\infty)\times\rr$ for which the limit superior
\begin{equation}\label{conv_cond}
\underset{\epsilon\downarrow0}{\overline{\lim}}\;
\bigg|\int_\rr e^{-\frac{z^2}{2t}}\,\sum_{k=0}^\infty \frac{(-1)^k\,z^{2k}}{(2k)!\,2^k\,t^k}\,
\bigg(\frac{\mathrm{d}}{\mathrm{d}\xi}\bigg)^{2k}\,\int_0^t \int_\rr \tilde{\eta}^{\epsilon}(s,\xi-\chi,y)\,s^{-1/2}
\,e^{-\frac{\chi^2}{2s}}\,\mathrm{d}\chi\,\mathrm{d}s\,\mathrm{d}z\bigg|
\end{equation}
is finite, the error bound
\begin{equation}\label{conv_fast_st}
\underset{\epsilon\downarrow0}{\overline{\lim}}\;\,\epsilon^{-1}\,
\Big|V(t,x,y)-V^{(0)}(t,x)-\sqrt{\epsilon}\,V^{(1)}(t,x)\Big|<\infty
\end{equation}
applies. If the limit superior \eqref{conv_cond} is bounded above uniformly on a subset of $[0,T)\times(0,\infty)\times\rr$, then the convergence in \eqref{conv_fast_st} is uniform on the same subset of $[0,T)\times(0,\infty)\times\rr$. 
\end{enumerate}
\end{theorem}  

\begin{rmk}
Condition \eqref{conv_cond} is of the same form as condition \eqref{etacontr2}, and the detailed interpretation of the latter given in Remark \ref{conv_rmk} applies here as well. 
\end{rmk}

\smallskip

\noindent\textit{Proof of Theorem \ref{fast_conv_thm}.} We proceed as in the proof of Theorem \ref{conv_prop}. Specifically, we insert the ansatz $V=V^{(0)}+\epsilon^{1/2}\,V^{(1)}+\epsilon\,V^{(2)}+\epsilon^{3/2}\,V^{(3)}+\epsilon^2\,Q$ into the HJB equation \eqref{HJB4} and expand the resulting equation in the powers of $\epsilon^{1/2}$. The terms $V^{(0)}$, $V^{(1)}$, $V^{(2)}$, and $V^{(3)}$ were chosen in such a way that all terms on the orders of $\frac{1}{\epsilon}$, $\frac{1}{\epsilon^{1/2}}$, $1$ and $\epsilon^{1/2}$ cancel out. At this point, a tedious but straightforward computation relying on the elementary identity 
\begin{equation}
\frac{1}{a+\epsilon^{1/2}\,b}=\frac{1}{a}-\epsilon^{1/2}\,\frac{b}{a^2+\epsilon^{1/2}\,a\,b},\qquad a<0,\;\;b<-\epsilon^{-1/2}\,a
\end{equation}
allows to compute the terms of order $\epsilon$ and leads to 
\begin{equation}\label{R_fast_eq}
\epsilon\,\tilde{Q}_t+\frac{\epsilon}{2}\,\frac{\|\lambda(y)\|^2\,\big(V^{(0)}_x\big)^2}{\big(V^{(0)}_{xx}\big)^2}\,\tilde{Q}_{xx}
-\epsilon\,\frac{\|\lambda(y)\|^2\,V^{(0)}_x}{V^{(0)}_{xx}}\,\tilde{Q}_x = \epsilon\,\eta^{\epsilon},
\end{equation}
where $\tilde{Q}:=\epsilon^{-1}\big(V-V^{(0)}-\epsilon^{1/2}\,V^{(1)}\big)=V^{(2)}+\epsilon^{1/2}\,V^{(3)}+\epsilon\,Q$ and $\eta^{\epsilon}$ is defined according to \eqref{eta_fast}. One can then conclude the argument by repeating the steps in the proof of Theorem \ref{conv_prop}, making the change of coordinates
\begin{equation}
(t,\xi,y):=\Big(t,-\log V^{(0)}_x-\frac{\|\lambda(y)\|^2}{2}\,t,y\Big)
\end{equation}
in \eqref{R_fast_eq}, and combining Duhamel's principle for the resulting equation with the formula for the inverse Weierstrass transform given in \cite{Wi2}. \ep 

\begin{rmk}\label{fast_ex_rmk}
Consider the example of Section \ref{new_ex} with the explicit solution reparametrized according to $\delta=\epsilon^{-1}$, and take the corresponding approximation for a fast volatility factor as considered in this section. Then, by direct computation $\phi(y)=\|\Lambda\|^2(y-m)$ and, therefore, our approximation is
\eq
V^{(0)}+\sqrt{\epsilon}\,V^{(1)}=\gamma^\gamma\,\frac{x^{1-\gamma}}{1-\gamma}\,e^{-\frac{1}{2}\,\Gamma\,\bar{\lambda}^2 t}+\sqrt{\epsilon}\,\|\Lambda\|^2\,\beta\,m\,\Lambda^T\rho\\,\Gamma^2\,\gamma^\gamma\,\frac{x^{1-\gamma}}{1-\gamma}\,te^{-\frac{1}{2}\,\Gamma\,\bar{\lambda}^2 t}.
\en
The corresponding remainder $\eta^\epsilon$ has a complicated dependence on $\epsilon$. However, an explicit Taylor expansion in $\epsilon$ of the functions $A_1$, $A_2$ in the formula for $V$ of \eqref{V_example} (with $\delta$ replaced by $\epsilon^{-1}$) shows that $\eta^\epsilon$ is given by a product of $x^{1-\gamma}$ and an order one in $\epsilon$ function of $t$, $y$. Consequently, $\tilde{\eta}^\epsilon$ is a product of $e^{\xi\Gamma}$ and a function of $(t,y)$ that is of order $1$ in $\epsilon$. The argument of Section \ref{new_ex} yields that the conditions of Theorem \ref{fast_conv_thm} are satisfied in this example and its fast factor approximation.  
  
\end{rmk}

\section{Multiscale forward investment problem}\label{sec_multi}

We combine our approaches to the forward investment problems with slow and fast factors to analyze the multiscale forward investment problem described in Section \ref{sub_setting}. We consider an expansion for $V(t,x,y_1,y_2)$ in equation \eqref{HJB1} of the form
\begin{equation}
V=V^{(0)}+\sqrt{\delta}\,V^{(1,0)}+\sqrt{\epsilon}\,V^{(0,1)}+O(\delta+\epsilon)
\end{equation}
in the limit regime $\delta\downarrow0$, $\epsilon\downarrow0$. We first give the general results and, in turn, explicit formulas for the case of power utilities in Section \ref{power}.

\subsection{First Order Approximations}
It is convenient to define:
\begin{align}
\bar{\lambda}(y_1) & = \left(\int_\rr \big\|\lambda(y_1,y_2)\big\|^2\,\mu(\mathrm{d}y_2)\right)^{1/2}, \label{lbareqn}\\
C_{1,0}(y_1) & =  \big(\rho^{\mathbf{s}}\big)^T\left(\int_\rr \lambda(y_1,y_2)\,\mu(\mathrm{d}y_2)\right)
\kappa(y_1), \label{C10eqn}\\
C_{0,1}(y_1) & =  \big(\rho^{\mathbf{f}}\big)^T\left(\int_\rr \lambda(y_1,y_2)\,\phi_{y_2}(y_1,y_2)\,
\alpha(y_2)\,\mu(\mathrm{d}y_2)\right), \label{C01eqn}
\end{align}
where 
\begin{equation}
\phi(y_1,y_2)=\int_0^\infty \ev\Big[\big\|\lambda(y_1,Y^{1}(s))\big\|^2
-\bar{\lambda}^2(y_1)\mid Y^1(0)=y_2\Big]\,\mathrm{d}s, \label{phidef}
\end{equation}
and $Y^{1}$ denotes the fast factor process, solution of the SDE \eqref{Yeps}, but with $\epsilon=1$.

The following proposition gives explicit formulas for the leading order term $V^{(0)}$ and the first order correction terms $V^{(1,0)}$ and $V^{(0,1)}$.

\begin{proposition}[Explicit formulas, general case]\label{multi_prop}
\begin{enumerate}[(i)]
\item The leading order term $V^{(0)}$ admits the representation
\begin{equation}
V^{(0)}(t,x,y_1,y_2)=V^{(0)}(t,x,y_1)=u\left(\bar{\lambda}^2(y_1)t,x\right) 
\end{equation}
where $u$ is given by
\begin{eqnarray*}
&& u(t,x)=-\frac{1}{2}\,\int_0^t e^{-h^{(-1)}(s,x)+\frac{s}{2}}\,h_x\big(s,h^{(-1)}(s,x)\big)\,\mathrm{d}s
+V(0,x), \\
&& h(t,x)=\int_\rr \frac{e^{zx-\frac{1}{2}z^2t}-1}{z}\,\nu_0(\mathrm{d}z). +C_0.
\end{eqnarray*}
Here $h^{(-1)}$ is the inverse of $h$ in the variable $x$, and $\nu_0$ 
and $C_0$ were introduced in Assumption \ref{main_asmp}(iii). 
\smallskip

\item The slow scale correction term $V^{(1,0)}$ is given by
\begin{equation}
V^{(1,0)}(t,x,y_1)=\frac{t}{2}\,C_{1,0}(y_1)\frac{V^{(0)}_{xy_{1}}\,V^{(0)}_x}{V^{(0)}_{xx}} \label{V10}
\end{equation}
and admits the natural parametrization 
\begin{equation}\label{V10formula}
V^{(1,0)}(t,x,y_1)=\int_0^t V^{(1),\delta,s}(t,x,y_1)\,\mathrm{d}s,
\end{equation}
where each $V^{(1),\delta,s}$ is a solution of the initial value problem
\[
V^{(1),\delta,s}_t + \frac{\bar{\lambda}^2(y_1)}{2}\,
\left(\frac{V^{(0)}_x}{V^{(0)}_{xx}}\right)^2 V^{(1),\delta,s}_{xx} -\bar{\lambda}^2(y_1)
\left(\frac{V^{(0)}_x}{V^{(0)}_{xx}}\right)V^{(1),\delta,s}_x=0,\quad t\ge s\,,
\]
with initial condition
\[
V^{(1),\delta,s}(s,x,y_1)=C_{1,0}(y_1)\,\frac{V^{(0)}_{xy_1}(s,x,y_1)V^{(0)}_x(s,x,y_1)}{V^{(0)}_{xx}(s,x,y_1)}\,.
\]
\smallskip
\item The fast scale correction term $V^{(0,1)}$ is given by
\begin{equation}
V^{(0,1)}(t,x,y_1)=t\,C_{0,1}(y_1)
\left(\frac{V^{(0)}_x}{V^{(0)}_{xx}}\right)
\bigg(\frac{\big(V^{(0)}_x\big)^2}{V^{(0)}_{xx}}\bigg)_x. \label{V01}
\end{equation}
The function $V^{(0,1)}$ admits the natural parametrization
\begin{equation}
V^{(0,1)}(t,x,y_1)=\int_0^t V^{(1),\epsilon,s}(t,x,y_1)\,\mathrm{d}s,
\end{equation}
where each $V^{(1),\epsilon,s}$ is a solution of the initial value problem
\[ 
V^{(1),\epsilon,s}_t + \frac{\bar{\lambda}^2(y_1)}{2}\,
\left(\frac{V^{(0)}_x}{V^{(0)}_{xx}}\right)^2\,V^{(1),\epsilon,s}_{xx} -\bar{\lambda}^2(y_1)
\left(\frac{V^{(0)}_x}{V^{(0)}_{xx}}\right)V^{(1),\epsilon,s}_x=0,\quad t\ge s\,, 
\]
with initial condition
\[
V^{(1),\epsilon,s}(s,x,y_1)=C_{0,1}(y_1)\,\frac{V^{(0)}_x(s,x,y_1)}{V^{(0)}_{xx}(s,x,y_1)}\bigg(\frac{V^{(0)}_x(s,x,y_1)^2}{V^{(0)}_{xx}(s,x,y_1)}\bigg)_x.
\]
\end{enumerate}
\end{proposition}

\begin{rmk}\label{multirmk}
The quantities $V^{(1),\delta,s}$, $V^{(1),\epsilon,s}$ of Proposition \ref{multi_prop} should be interpreted in the same way as their analogues in the single factor cases. We refer to Remark \ref{rmk_w1s} for more details. However, we highlight that the analysis identifies the following reduced parameters: $\bar{\lambda}(y_1)$ in \eqref{lbareqn}, the Sharpe ratio root mean square-averaged with respect to the fast factor and frozen at the value $y_1$ of the slow factor; $C_{1,0}(y_1)$ in \eqref{C10eqn}, which has the effect of the correlation $\rho^{\mathbf{s}}$ between the slow factor and equity returns; and $C_{0,1}(y_1)$ in \eqref{C01eqn}, which has the effect of the correlation $\rho^{\mathbf{f}}$ between the fast factor and equity returns.
\end{rmk}

\smallskip

\noindent\textit{Proof of Proposition \ref{multi_prop}.} We start with the proofs of parts (i) and (iii). To this end, we insert $\delta=0$ into the HJB equation \eqref{HJB1}, employ Assumption \ref{main_asmp}, and then proceed as in the proofs of Propositions \ref{ffV0prop}, \ref{fast_corr_prop}, and \ref{fastnp}. The arguments from therein can be repeated directly by replacing $\lambda(y)$ by $\lambda(y_1,y_2)$. In particular, $V^{(0)}$ and $V^{(0,1)}$ are determined via an averaging of the equations
\eq
V^{(0)}_t-\frac{\|\lambda\|^2}{2}\,\frac{\big(V^{(0)}_x\big)^2}{V^{(0)}_{xx}}
+\mathcal{L}_{y_2}\,V^{(0,2)}=0 \label{LV2}
\en
and
\eq
V^{(0,1)}_t+\frac{\|\lambda\|^2}{2}\left(\frac{V^{(0)}_x}{V^{(0)}_{xx}}\right)^2 V^{(0,1)}_{xx}
-\frac{V^{(0)}_x}{V^{(0)}_{xx}}\lambda^T\bigg(V^{(0,1)}_x\,\lambda
+\phi_{y_2}\,\bigg(\frac{\big(V^{(0)}_x\big)^2}{V^{(0)}_{xx}}\bigg)_x\,\alpha\,\rho\bigg)
+\mathcal{L}_{y_2}\,V^{(0,3)}=0 \label{LV3}
\en
with respect to $\mu(\mathrm{d}y_2)$. Here $\epsilon^{-1}\mathcal{L}_{y_2}$ is the generator of the fast factor $Y^{\epsilon}$, that is
\[ \mathcal{L}_{y_2} = \frac12\alpha(y_2)^2\partial_{y_2y_2} + \gamma(y_2)\partial_{y_2}, \]
and the terms $V^{(0)}$, $V^{(0,1)}$, $V^{(0,2)}$, $V^{(0,3)}$ are the ones appearing in the expansion of the solution to \eqref{HJB1} with $\delta=0$. In particular, subtracting from \eqref{LV2} its averaged version, we obtain the expression
\begin{equation}
V^{(0,2)}(t,x,y_1,y_2) = -\frac12\phi(y_1,y_2)\frac{\big(V^{(0)}_x\big)^2}{V^{(0)}_{xx}}
+C(t,x,y_1) \label{V02eqn}
\end{equation}
where $C(t,x,y_1)$ is a function that does not depend on $y_2$, and $\phi$ was defined in \eqref{phidef}.

\medskip

It remains to prove part (ii). To this end, we again employ Assumption \ref{main_asmp} and insert the ansatz $V^{(0)}+\sqrt{\delta}\,V^{(1)}$ into the HJB equation \eqref{HJB1}. Collecting the terms of order $\sqrt{\delta}$ in the resulting equation we get
\begin{equation}\label{V1mixed}
\begin{split}
V^{(1)}_t-\bigg(\lambda\,\frac{V^{(0)}_x}{V^{(0)}_{xx}}
+\frac{\big(\rho^{\mathbf{f}}\big)^T\alpha}{\sqrt{\epsilon}}\,\frac{V^{(0)}_{xy_2}}{V^{(0)}_{xx}}\bigg)
\bigg(\rho^{\mathbf{s}}\kappa\,V^{(0)}_{xy_1}+\lambda\,V^{(1)}_x
+\frac{\rho^{\mathbf{f}}\alpha}{\sqrt{\epsilon}}\,\frac{V^{(1)}_{xy_2}}{V^{(0)}_{xx}}\bigg)
+\frac{1}{2}\,\frac{\big\|\lambda\,V^{(0)}_x
+\frac{\rho^{\mathbf{f}}\,\alpha}{\sqrt{\epsilon}}\,V^{(0)}_{xy_2}\big\|^2\,V^{(1)}_{xx}}{V^{(0)}_{xx}}\\
+\frac{1}{\epsilon}\,\gamma(y)\,V^{(1)}_{y_2}+\frac{1}{2\epsilon}\,\alpha\,V^{(1)}_{y_2y_2}
+\frac{1}{\sqrt{\epsilon}}\,\kappa\,\alpha\,\rho^{\mathbf{s,f}}\,V^{(0)}_{y_1y_2}=0.
\end{split}
\end{equation}
We can now write $V^{(1)}=V^{(1,0)}+\sqrt{\epsilon}\,V^{(1,1)}+\epsilon\,V^{(1,2)}$ and expand equation \eqref{V1mixed} in powers of $\epsilon$ as in the proof of Proposition \ref{fast_corr_prop}. By doing so, we conclude that $V^{(1,0)}$ and $V^{(1,1)}$ can be chosen as functions independent of $y_2$. 

\medskip

Moreover, $V^{(1,0)}$ can be determined by averaging the equation
\begin{equation}\label{V1slmixed}
\begin{split}
V^{(1,0)}_t
+\frac{1}{2}\,\|\lambda\|^2\left(\frac{V^{(0)}_x}{V^{(0)}_{xx}}\right)^2 V^{(1,0)}_{xx}
-\|\lambda\|^2\,\frac{V^{(0)}_x}{V^{(0)}_{xx}}\,V^{(1,0)}_x 
-\kappa\,\lambda^T\rho^{\mathbf{s}}\,\frac{V^{(0)}_{xy_1}\,V^{(0)}_x}{V^{(0)}_{xx}}
+\mathcal{L}_{y_2}\,V^{(1,2)}=0
\end{split}
\end{equation}
with respect to $\mu(\mathrm{d}y_2)$. The averaged equation is endowed with the initial condition $V^{(1,0)}(0,x,y_1)=0$ and can be solved explicitly by means of a transformation to an ill-posed backward heat equation as in the proofs of Propositions \ref{CorrTerm} and \ref{slownpprop}. This gives the explicit formula for $V^{(1,0)}$ and its natural parametrization. \ep

\begin{rmk}\label{overlaps3}
Formulas \eqref{V10formula} and \eqref{V01} can be considered as the forward performance analogs of those appearing in \cite[Section 4.1]{FSZ} for the (backwards in time) Merton problem. The differences highlighted in Remark \ref{overlaps} apply here too.
\end{rmk}
\medskip

We now complement the explicit formulas for the leading order terms $V^{(0)}$, $V^{(1,0)}$, $V^{(0,1)}$ by a convergence theorem justifying the approximation of the true value function $V$ by the function $V^{(0)}+\sqrt{\delta}\,V^{(1,0)}+\sqrt{\epsilon}\,V^{(0,1)}$. We will need the non-linear functional $\eta^{\delta,\epsilon}$, whose lengthy formula we give in Appendix \ref{Appeta}. 
We also set $(t,\xi,y_1,y_2):=\big(t,-\log V^{(0)}_x-\frac{\|\lambda\|^2}{2}\,t,y_1,y_2\big)$, and let 
\begin{equation}\label{etatildemulti}
\tilde{\eta}^{\delta,\epsilon}(t,\xi,y_1,y_2):=\eta^{\delta,\epsilon}(t,x,y_1,y_2).
\end{equation}

\begin{theorem}[Remainder estimate, general case]\label{multi_conv_thm}
Suppose that there are $\delta_0>0$, $\epsilon_0>0$, and $T\le\infty$ such that, for all $(\delta,\epsilon)\in(0,\delta_0)\times(0,\epsilon_0)$, the HJB equation \eqref{HJB1} has a solution $V\in C^{1,2,2,2}\big([0,T)\times(0,\infty)\times\rr^2\big)$ which is increasing and strictly concave in the second argument. Then, 
\begin{enumerate}[(i)]
\item the quantity
\begin{equation}
(\delta+\epsilon)^{-1}\,\int_\rr e^{-\frac{z^2}{2t}}\,\sum_{k=0}^\infty \frac{(-1)^k\,z^{2k}}{(2k)!\,2^k\,t^k}\,
\bigg(\frac{\mathrm{d}}{\mathrm{d}\xi}\bigg)^{2k}\,\int_0^t \int_\rr \tilde{\eta}^{\delta,\epsilon}(s,\xi-\chi,y_1,y_2)\,s^{-1/2}
\,e^{-\frac{\chi^2}{2s}}\,\mathrm{d}\chi\,\mathrm{d}s\,\mathrm{d}z,
\end{equation}
with $\tilde{\eta}^{\delta,\epsilon}$ defined in \eqref{etatildemulti}, is well-defined and finite for all $(\delta,\epsilon)\in(0,\delta_0)\times(0,\epsilon_0)$, and 
\item for every $(t,x,y_1,y_2)\in[0,T)\times(0,\infty)\times\rr^2$ for which the limit superior
\begin{equation}\label{conv_cond_multi}
\underset{\delta\downarrow0,\epsilon\downarrow0}{\overline{\lim}}\;(\delta+\epsilon)^{-1}
\bigg|\int_\rr e^{-\frac{z^2}{2t}}\,\sum_{k=0}^\infty \frac{(-1)^k\,z^{2k}}{(2k)!\,2^k\,t^k}\,
\bigg(\frac{\mathrm{d}}{\mathrm{d}\xi}\bigg)^{2k}\,\int_0^t \int_\rr \tilde{\eta}^{\delta,\epsilon}(s,\xi-\chi,y_1,y_2)\,s^{-1/2}
\,e^{-\frac{\chi^2}{2s}}\,\mathrm{d}\chi\,\mathrm{d}s\,\mathrm{d}z\bigg|
\end{equation}
is finite, the error bound
\begin{equation}\label{conv_multi_st}
\underset{\delta\downarrow0,\epsilon\downarrow0}{\overline{\lim}}\;\,(\delta+\epsilon)^{-1}\;
\Big|V(t,x,y_1,y_2)-V^{(0)}(t,x,y_1)-\sqrt{\delta}\,V^{(1,0)}(t,x,y_1)
-\sqrt{\epsilon}\,V^{(0,1)}(t,x,y_1)\Big|<\infty
\end{equation}
applies. If the limit superior \eqref{conv_cond_multi} is bounded above uniformly on a subset of $[0,T)\times(0,\infty)\times\rr^2$, then the convergence in \eqref{conv_multi_st} is uniform on the same subset of $[0,T)\times(0,\infty)\times\rr^2$. 
\end{enumerate}
\end{theorem}  

\begin{rmk}
Condition \eqref{conv_cond_multi} is of the same form as condition \eqref{etacontr2}, and the detailed interpretation of the latter given in Remark \ref{conv_rmk} applies here as well. 
\end{rmk}

\smallskip

\noindent\textit{Proof of Theorem \ref{multi_conv_thm}.} We proceed as in the proof of Theorem \ref{conv_prop}. More specifically, we plug $V=V^{(0)}+\sqrt{\delta}\,V^{(1,0)}+\sqrt{\epsilon}\,V^{(0,1)}+Q$ into the HJB equation \eqref{HJB1} and apply a Taylor's expansion to the resulting equation in $\sqrt{\delta}$ and $\sqrt{\epsilon}$. Hereby, we use the elementary identity
\begin{equation}
\frac{1}{a+b}=\frac{1}{a}-\frac{b}{a^2+ab}
\end{equation}
and the definitions of $V^{(0)}$, $V^{(1,0)}$, $V^{(0,1)}$, $V^{(2)}$, $V^{(3)}$, $V^{(1,1)}$, and $V^{(1,2)}$ to eliminate the terms of orders $1$, $\sqrt{\delta}$, and $\sqrt{\epsilon}$. The remaining equation then reads 
\begin{equation}\label{R_multi_eq}
Q_t+\frac{\|\lambda\|^2}{2}\left(\frac{V^{(0)}_x}{V^{(0)}_{xx}}\right)^2 Q_{xx}
-\|\lambda\|^2\left(\frac{V^{(0)}_x}{V^{(0)}_{xx}}\right)\,Q_x = \eta^{\delta,\epsilon},
\end{equation}
where $\eta^{\delta,\epsilon}$ is defined prior to the statement of the theorem. One can now conclude by repeating the steps in the proof of Theorem \ref{conv_prop}, namely by making the change of coordinates
\begin{equation}\label{coc_multi}
(t,\xi,y_1,y_2):=\Big(t,-\log V^{(0)}_x-\frac{\|\lambda\|^2}{2}\,t,y_1,y_2\Big)
\end{equation}
in \eqref{R_multi_eq}, and combining Duhamel's principle for the resulting equation with the formula for the inverse Weierstrass transform given in \cite{Wi2}. \ep 

\subsection{Power utility example\label{power}}
We illustrate the results with the family of power utility forward performance processes. For a constant risk aversion coefficient $\gamma\in(0,\infty)\backslash\{1\}$, we impose the initial condition of the HJB equation \eqref{HJB1} to be
\[ V(0,x) = \gamma^\gamma\,\frac{x^{1-\gamma}}{1-\gamma}\,. \]
Similar to Section \ref{new_ex}, we have 
the following explicit solution for the constant parameter value function $V^{(0)}(t,x,y_1)$ in part (i) of Proposition \ref{multi_prop}:
\begin{equation}
V^{(0)}(t,x,y_1)=u\left(\bar{\lambda}^2(y_1)t,x\right) \quad \mbox{where} \quad u(t,x)=\gamma^\gamma\,
\frac{x^{1-\gamma}}{1-\gamma}e^{-\frac{1}{2}\Gamma t}, \quad\mbox{and}\quad \Gamma=\frac{1-\gamma}{\gamma},\label{V0power}
\end{equation}
which can be verified by taking the measure $\nu_0$ to be a Dirac delta centered at $\gamma^{-1}$ and the constant $C_0=\gamma$. 

From \eqref{V10} in Proposition \ref{multi_prop} we compute
\[ 
V^{(1,0)}(t,x,y_1)=\frac12 t^2C_{1,0}(y_1)\Gamma^2\bar\lambda(y_1)\bar\lambda'(y_1)
V^{(0)}(t,x,y_1), 
\]
and from \eqref{V01} we obtain
\[ V^{(0,1)}(t,x,y_1)=t\,C_{0,1}(y_1)\Gamma^2V^{(0)}(t,x,y_1). \]
Then, the three-term approximation to the forward performance value function is given by
\begin{equation}
 V(t,x,y_1,y_2)=\left(1+\frac12 \sqrt{\delta}\,t^2C_{1,0}(y_1)\Gamma^2\bar\lambda(y_1)\bar\lambda'(y_1) + \sqrt{\epsilon}\,t\,C_{0,1}(y_1)\Gamma^2\right)V^{(0)}(t,x,y_1) + O(\delta+\epsilon), \label{Vpowerapp}
 \end{equation}
where $V^{(0)}$ is given explicitly in \eqref{V0power}.

\section{Approximately optimal portfolio} \label{sec_port}

In this last section we define the portfolio associated with our approximation and establish its approximate optimality. 

\begin{definition}\label{opt_def}
Let the feedback portfolio function $\pi^\approx$ be given by 
\begin{equation}\label{opt_port}
\begin{split}
\pi^\approx=-(\sigma^T)^{-1}\,\lambda\,\bigg(\frac{V^{(0)}_x}{V^{(0)}_{xx}}+\sqrt{\delta}\,\frac{V^{(0)}_{xx}\,V^{(1,0)}_x-V^{(0)}_x\,V^{(1,0)}_{xx}}{\big(V^{(0)}_{xx}\big)^2}+\sqrt{\epsilon}\,\frac{V^{(0)}_{xx}\,V^{(0,1)}_x-V^{(0)}_x\,V^{(0,1)}_{xx}}{\big(V^{(0)}_{xx}\big)^2}\bigg) \\
-\sqrt{\delta}\,\kappa\,(\sigma^T)^{-1}\,\rho^{\mathbf{s}}\,\frac{V^{(0)}_{xy_1}}{V^{(0)}_{xx}} 
+ \sqrt{\epsilon}\,(\sigma^T)^{-1}\rho^{\mathbf{f}}\alpha 
\frac{\phi_{y_2}}{V^{(0)}_{xx}}\bigg(\frac{\big(V^{(0)}_x\big)^2}{V^{(0)}_{xx}}\bigg)_x
\end{split}
\end{equation}
where $V^{(0)}$, $V^{(1,0)}$ and $V^{(0,1)}$ are as in Proposition \ref{multi_prop} and $\phi$ is given in \eqref{phidef}. The formula \eqref{opt_port} is obtained by recalling that the nonlinearity in the HJB equation \eqref{HJB1} results from the optimization problem
\[
\sup_\pi \Big((\lambda\,V_x+\sqrt{\delta}\,\kappa\,\rho^{\mathbf{s}}\,V_{xy_1}+\frac{1}{\sqrt{\epsilon}}\,\alpha\,\rho^{\mathbf{f}}\,V_{xy_2})^T(\sigma\,\pi)+\frac{1}{2}\,V_{xx}\,(\sigma\,\pi)^T(\sigma_\pi)\Big),
\]
replacing $V$ by its expansion $V^{(0)}+\sqrt{\delta}\,V^{(1,0)}+\sqrt{\epsilon}\,V^{(0,1)}$ in the formula for the corresponding optimizer $\pi^*$, and applying Taylor's expansion in $\delta$, $\epsilon$ to the result. We refer to $\pi^\approx$ as the \textnormal{approximately optimal portfolio} which is justified by the next proposition. 
\end{definition}

\begin{rmk}
One can use formula \eqref{Vpowerapp} for the value function in the case of the power forward performance to compute the approximately optimal portfolio. We omit the lengthy expression here.
\end{rmk}

\begin{proposition}\label{opt_prop}
Suppose that there are $\delta_0>0$, $\epsilon_0>0$, and $T\le\infty$ such that, for all $(\delta,\epsilon)\in(0,\delta_0)\times(0,\epsilon_0)$, the HJB equation \eqref{HJB1} has a solution $V^{\delta,\epsilon}\in C^{1,2,2,2}([0,T)\times(0,\infty)\times\rr^2)$ which is increasing and strictly concave in the second argument. Then, the value process $V^{\delta,\epsilon}(t,X^{\pi^\approx}(t),Y^\delta(t),Y^\epsilon(t))$ satisfies a SDE of the form 
\eq\label{V_SDE_approx}
\mathrm{d}V^{\delta,\epsilon} = V^{\delta,\epsilon}_{y_1}\,\sqrt{\delta}\,\kappa\,\mathrm{d}B_1(t)+V^{\delta,\epsilon}_{y_2}\,\frac{1}{\sqrt{\epsilon}}\,\alpha\,\mathrm{d}B_2(t)+V^{\delta,\epsilon}_x\,\sigma^T\pi^\approx\,\mathrm{d}W(t)+\Theta^{\delta,\epsilon}\,\mathrm{d}t
\en
with drift coefficient $\Theta^{\delta,\epsilon}=\Theta^{\delta,\epsilon}(t,X^{\pi^\approx}(t),Y^\delta(t), Y^\epsilon(t))$. If, in addition, the limits superior resulting from \eqref{conv_cond_multi} by replacing $\tilde{\eta}^{\delta,\epsilon}$ with any of $\tilde{\eta}^{\delta,\epsilon}_\xi$, $\tilde{\eta}^{\delta,\epsilon}_{\xi\xi}$, $\tilde{\eta}^{\delta,\epsilon}_{\xi\xi\xi}$, $\tilde{\eta}^{\delta,\epsilon}_{\xi y_1}$, $\tilde{\eta}^{\delta,\epsilon}_{\xi y_2}$ are finite, 
then   
\eq\label{Theta_cont}
\underset{\delta\downarrow0,\epsilon\downarrow0}{\overline{\lim}}\;(\delta+\epsilon)^{-1}\,\big|\Theta^{\delta,\epsilon}(t,x,y_1,y_2)\big|<\infty.
\en
In other words, the performance of the portfolio $\pi^\approx$ fails to fulfill the martingale criterion of optimality only by a bounded variation term of order $\delta+\epsilon$.
\end{proposition} 

\begin{proof}
The SDE \eqref{V_SDE_approx} is obtained by combining
\begin{equation}
\mathrm{d}X^\pi(t)=\mu\big(Y^{\delta}(t),Y^{\epsilon}(t)\big)^T\pi^\approx(t)\,\mathrm{d}t
+\sigma\big(Y^{\delta}(t),Y^{\epsilon}(t)\big)^T\pi^\approx(t)\,\mathrm{d}W(t),
\end{equation} 
the definition of $\pi^\approx$ (Definition \ref{opt_def}), and It\^o's formula. Writing $V^{\delta,\epsilon}=V^{(0)}+\sqrt{\delta}\,V^{(1,0)}+\sqrt{\epsilon}\,V^{(0,1)}+Q^{\delta,\epsilon}$ and recalling the HJB equation \eqref{HJB1} we see that the drift coefficient $\Theta^{\delta,\epsilon}$ is a linear combination of $Q^{\delta,\epsilon}_x$, $Q^{\delta,\epsilon}_{xy_1}$, $Q^{\delta,\epsilon}_{xy_2}$, and $Q^{\delta,\epsilon}_{xx}$, with the coefficients being uniformly bounded in $\delta$, $\epsilon$ wherever the limits superior $\underset{\delta\downarrow0,\epsilon\downarrow0}{\overline{\lim}}\,V^{\delta,\epsilon}_x$, $\underset{\delta\downarrow0,\epsilon\downarrow0}{\overline{\lim}}\,\big|V^{\delta,\epsilon}_{xy_1}\big|$, $\underset{\delta\downarrow0,\epsilon\downarrow0}{\overline{\lim}}\,\big|V^{\delta,\epsilon}_{xy_2}\big|$, $\underset{\delta\downarrow0,\epsilon\downarrow0}{\overline{\lim}}\,-\frac{1}{V^{\delta,\epsilon}_{xx}}$ are finite. 

\medskip

Next, recall the change of variables \eqref{coc_multi}, write $q^{\delta,\epsilon}(t,\xi,y_1,y_2)$ for $Q^{\delta,\epsilon}(t,x,y_1,y_2)$, and recall from the proof of Theorem \ref{multi_conv_thm} that $q^{\delta,\epsilon}$ satisfies the backward equation
\eq\label{q_heat}
q^{\delta,\epsilon}_t+\frac{\|\lambda\|^2}{2}\,q^{\delta,\epsilon}_{\xi\xi}=\tilde{\eta}^{\delta,\epsilon}
\en
with zero initial condition. Differentiation of this equation in $\xi$, Duhamel's principle for the resulting equation, and the formula for the inverse Weierstrass transform given in \cite{Wi2} yield that $q^{\delta,\epsilon}_\xi$ (and therefore also $Q^{\delta,\epsilon}_x$) admits a control of the form \eqref{Theta_cont} at a point if condition \eqref{conv_cond_multi} holds for $\tilde{\eta}^{\delta,\epsilon}_\xi$ at that point. Repeated differentiation in $\xi$ gives the same result for $q^{\delta,\epsilon}_{\xi\xi}$ (and therefore also $Q^{\delta,\epsilon}_{xx}$) and $\tilde{\eta}^{\delta,\epsilon}_{\xi\xi}$, as well as for $q^{\delta,\epsilon}_{\xi\xi\xi}$ and $\tilde{\eta}^{\delta,\epsilon}_{\xi\xi\xi}$.

\medskip

Lastly, differentiate the equation \eqref{q_heat} in $\xi$, $y_1$ to get 
\eq
\big(q^{\delta,\epsilon}_{\xi y_1}\big)_t+\frac{\|\lambda\|^2}{2}\,\big(q^{\delta,\epsilon}_{\xi y_1}\big)_{\xi\xi}=\tilde{\eta}^{\delta,\epsilon}_{\xi y_1}-\lambda_{y_1}^T\,\lambda\,q^{\delta,\epsilon}_{\xi\xi\xi}. 
\en
By linearity $q^{\delta,\epsilon}_{\xi y_1}=\hat{q}^{\delta,\epsilon}-t\,\lambda_{y_1}^T\,\lambda\,q^{\delta,\epsilon}_{\xi\xi\xi}$, where $\hat{q}^{\delta,\epsilon}$ solves
\eq
\hat{q}^{\delta,\epsilon}_t+\frac{\|\lambda\|^2}{2}\,\hat{q}^{\delta,\epsilon}_{\xi\xi}=\tilde{\eta}^{\delta,\epsilon}_{\xi y_1}.
\en
Hence, $\hat{q}^{\delta,\epsilon}$ admits a control of the form \eqref{Theta_cont} wherever condition \eqref{conv_cond_multi} is satisfied by $\tilde{\eta}^{\delta,\epsilon}_{\xi y_1}$, whereas $q^{\delta,\epsilon}_{\xi\xi\xi}$ has already been controlled before. The desired bound on $Q^{\delta,\epsilon}_{xy_1}$ then follows from the bounds on $q^{\delta,\epsilon}_{\xi y_1}$, $q^{\delta,\epsilon}_{\xi\xi}$, and $q^{\delta,\epsilon}_\xi$. To finish the proof it remains to estimate $Q^{\delta,\epsilon}_{xy_2}$ in a similar manner and to combine the estimates on $Q^{\delta,\epsilon}_x$, $Q^{\delta,\epsilon}_{xy_1}$, $Q^{\delta,\epsilon}_{xy_2}$, and $Q^{\delta,\epsilon}_{xx}$.
\end{proof}

\begin{rmk}
In the case that only a slow factor is present (that is, in the setting of Section \ref{sec_slow}) a statement analogous to that of Proposition \ref{opt_prop} holds for the portfolio function
\begin{equation}\label{opt_port_slow}
\pi^\approx=-(\sigma^T)^{-1}\,\lambda\,\bigg(\frac{V^{(0)}_x}{V^{(0)}_{xx}}+\sqrt{\delta}\,\frac{V^{(0)}_{xx}\,V^{(1)}_x-V^{(0)}_x\,V^{(1)}_{xx}}{\big(V^{(0)}_{xx}\big)^2}\bigg)-\sqrt{\delta}\,\kappa\,(\sigma^T)^{-1}\,\rho\,\frac{V^{(0)}_{xy}}{V^{(0)}_{xx}}
\end{equation}
where $V^{(0)}$ and $V^{(1)}$ are as in Propositions \ref{slow_lead} and \ref{CorrTerm}, respectively.

\medskip

We easily deduce that a similar result holds for the case that only a fast factor is present (that is, in the setting of Section \ref{sec_fast}) for the portfolio function
\begin{equation}\label{opt_port_fast}
\pi^\approx=-(\sigma^T)^{-1}\,\lambda\,\bigg(\frac{V^{(0)}_x}{V^{(0)}_{xx}}+\sqrt{\epsilon}\,\frac{V^{(0)}_{xx}\,V^{(1)}_x-V^{(0)}_x\,V^{(1)}_{xx}}{\big(V^{(0)}_{xx}\big)^2}\bigg)+\sqrt{\epsilon}\,(\sigma^T)^{-1}\rho\alpha\frac{\phi_{y}}{V^{(0)}_{xx}}\bigg(\frac{\big(V^{(0)}_x\big)^2}{V^{(0)}_{xx}}\bigg)_x
\end{equation}
where $V^{(0)}$ and $V^{(1)}$ are as in Propositions \ref{ffV0prop} and \ref{fast_corr_prop}, respectively, and $\phi(y)$ is given in \eqref{lambdahat}.   
\end{rmk}

\section{Conclusion}
We have provided a convergent approximation for forward performance processes in a multifactor incomplete markets model, as well as for the corresponding optimal portfolio. Our approach is based on a perturbation analysis of the corresponding ill-posed HJB equation. The principal term in the approximation results from the appropriately averaged problem, whose solution is known from Widder's Theorem. The correction terms for fast and slow volatility factors can be computed explicitly in terms of this leading order term.

\medskip

We have also given explicit calculations in the case of power utility. The ease of the formulas provided in more general cases, as well as conditions for convergence, should allow future work to develop the financial implications of forward performance processes in realistic market environments.

\newpage

\appendix
\section{Expression for $\eta^{\delta,\epsilon}$ in Section \ref{sec_multi}\label{Appeta}}
\begin{equation*}
\begin{split}
& \eta^{\delta,\epsilon}:=-\frac{1}{V_{xx}}\bigg(\frac{\delta}{2}\,\|\lambda\|^2\,\big(V^{(1,0)}_x\big)^2
+\frac{\epsilon}{2}\,\|\lambda\|^2\,\big(V^{(0,1)}_x\big)^2
+\frac{\|\lambda\|^2}{2}
\,\big(V_x-V^{(0)}_x-\sqrt{\delta}\,V^{(1,0)}_x-\sqrt{\epsilon}\,V^{(0,1)}_x\big)^2 \\
&+\frac{\kappa^2}{2}\,\|\rho^{\mathbf{s}}\|^2\,\delta\,\big(V^{(0)}_{xy_1}\big)^2 
+\frac{\delta^2}{2}\,\kappa^2\,\|\rho^{\mathbf{s}}\|^2\,\big(V^{(1,0)}_{xy_1}\big)^2 
+\frac{\delta}{2}\,\epsilon\,\kappa^2\,\|\rho^{\mathbf{s}}\|^2\,\big(V^{(0,1)}_{xy_1}\big)^2 
+\frac{1}{2\epsilon}\,\alpha^2\,\|\rho^{\mathbf{f}}\|^2\,(V_{xy_2})^2 \\
&+\frac{\delta}{2}\,\kappa^2
\,\|\rho^{\mathbf{s}}\|^2\,\big(V_{xy_1}-V^{(0)}_{xy_1}-\sqrt{\delta}\,V^{(1,0)}_{xy_1} 
-\frac{\sqrt{\epsilon}}{2}\,V^{(0,1)}_{xy_1}\big)^2 
+\delta\,\kappa\,\lambda^T\rho^{\mathbf{s}}\,V^{(0)}_x\,V^{(1,0)}_{xy_1} \\
&+\kappa\,\lambda^T\rho^{\mathbf{s}}\,\sqrt{\delta\,\epsilon}\,V^{(0)}_x\,V^{(0,1)}_{xy_1}
+\sqrt{\delta\,\epsilon}\,\|\lambda\|^2\,V^{(1,0)}_x\,V^{(0,1)}_x 
+\kappa\,\lambda^T\rho^{\mathbf{s}}\,\delta\,V^{(0)}_{xy_1}\,V^{(1,0)}_x \\
& +\delta^{3/2}\,\kappa\,\lambda^T\rho^{\mathbf{s}}\,V^{(1,0)}_x\,V^{(1,0)}_{xy_1} 
+\delta\,\sqrt{\epsilon}\,\kappa\,\lambda^T\rho^{\mathbf{s}}\,V^{(1,0)}_x\,V^{(0,1)}_{xy_1}  
+\frac{\sqrt{\delta}}{\sqrt{\epsilon}}\,\alpha\,\lambda^T\rho^{\mathbf{f}}\,V^{(1,0)}_x\,V_{xy_2} \\
& +\sqrt{\delta\,\epsilon}\,\kappa\,\lambda^T\rho^{\mathbf{s}}\,V^{(0)}_{xy_1}\,V^{(1,0)}_x 
+\delta\,\sqrt{\epsilon}\,\kappa\,\lambda^T\rho^{\mathbf{s}}\,V^{(0,1)}_x\,V^{(1,0)}_{xy_1} 
+\sqrt{\delta}\,\epsilon\,\kappa\,\lambda^T\rho^{\mathbf{s}}\,V^{(0,1)}_x\,V^{(0,1)}_{xy_1} \\
& +\alpha\,\lambda^T\rho^{\mathbf{f}}\,V^{(0,1)}_x\,V_{xy_2}
+\delta^{3/2}\,\kappa^2\,\|\rho^{\mathbf{s}}\|^2\,V^{(0)}_{xy_1}\,V^{(1,0)}_{xy_1}
+\delta\,\sqrt{\epsilon}\,\kappa^2\,\|\rho^{\mathbf{s}}\|^2\,V^{(0)}_{xy_1}\,V^{(0,1)}_{xy_1} \\
& +\frac{\sqrt{\delta}}{\sqrt{\epsilon}}\,\alpha\,\kappa\,\big(\rho^{\mathbf{s}}\big)^T\rho^{\mathbf{f}}\,
V^{(0)}_{xy_1}\,V_{xy_2} 
+\delta^{3/2}\,\sqrt{\epsilon}\,\kappa^2\,\|\rho^{\mathbf{s}}\|^2\,
V^{(1,0)}_{xy_1}\,V^{(0,1)}_{xy_1} 
+\frac{\delta}{\sqrt{\epsilon}}\,\alpha\,\kappa\,\big(\rho^{\mathbf{s}}\big)^T\rho^{\mathbf{f}}\,
V^{(1,0)}_{xy_1}\,V_{xy_2} \\
& +\sqrt{\delta}\,\alpha\,\kappa\,\big(\rho^{\mathbf{s}}\big)^T\rho^{\mathbf{f}}\,
V^{(0,1)}_{xy_1}\,V_{xy_2} 
+\kappa\,\rho^{\mathbf{s}}\,\big(V_{xy_1}-V^{(0)}_{xy_1}-\sqrt{\delta}\,V^{(1,0)}_{xy_1}
-\sqrt{\epsilon}\,V^{(0,1)}_{xy_1}\big) \\
&\quad\times\big(\sqrt{\delta}\,\lambda\,V^{(0)}_x+\sqrt{\delta}\,\lambda\,V^{(1,0)}_x
+\delta\,\kappa\,\rho^{\mathbf{s}}\,V^{(0)}_{xy_1}
+\delta^{3/2}\,\kappa\,\rho^{\mathbf{s}}\,V^{(1,0)}_{xy_1}
+\delta\,\sqrt{\epsilon}\,\kappa\,\rho^{\mathbf{s}}\,V^{(0,1)}_{xy_1}
+\frac{\sqrt{\delta}}{\sqrt{\epsilon}}\,\alpha\,\rho^{\mathbf{f}}\,V_{xy_2}\big) \\
& +\lambda\,\big(V_x-V^{(0)}_x-\sqrt{\delta}\,V^{(1,0)}_x-\sqrt{\epsilon}\,V^{(0,1)}_x\big) \\
& \quad \times\big(\sqrt{\delta}\,\lambda\,V^{(1,0)}_x+\sqrt{\epsilon}\,\lambda\,V^{(0,1)}_x
+\sqrt{\delta}\,\kappa\,\rho^{\mathbf{s}}\,V^{(0)}_{xy_1}
+\delta\,\kappa\,\rho^{\mathbf{s}}\,V^{(1,0)}_{xy_1} 
+\sqrt{\delta\,\epsilon}\,\kappa\,\rho^{\mathbf{s}}\,V^{(1,0)}_{xy_1}
+\frac{1}{\sqrt{\epsilon}}\,\alpha\,\rho^{\mathbf{f}}\,V_{xy_2}\big) 
\end{split}
\end{equation*}
\begin{equation*}
\begin{split}
& +\sqrt{\delta}\,\kappa\,\lambda^T\rho^{\mathbf{s}}
\big(V_x-V^{(0)}_x-\sqrt{\delta}\,V^{(1,0)}_x-\sqrt{\epsilon}\,V^{(0,1)}_x\big)
\,\big(V_{xy_1}-V^{(0)}_{xy_1}-\sqrt{\delta}\,V^{(1,0)}_{xy_1}-\sqrt{\epsilon}\,V^{(0,1)}_{xy_1}\big)
\bigg) \\
& -\bigg(\frac{1}{V_{xx}}-\frac{1}{V^{(0)}_{xx}}\bigg)
\,\bigg(\|\lambda\|^2\,\sqrt{\delta}\,V^{(0)}_x\,V^{(0),\delta}_x
+\|\lambda\|^2\,\sqrt{\epsilon}\,V^{(0)}_x\,V^{(0),\epsilon}_x
+\sqrt{\delta}\,\kappa\,\lambda^T\rho^{\mathbf{s}}\,V^{(0)}_x\,V^{(0)}_{xy_1} \\
&\quad\quad\quad\quad\quad\quad\quad\quad
+\frac{1}{\sqrt{\epsilon}}\,\alpha\,\lambda^T\rho^{\mathbf{f}}\,V^{(0)}_x\,V_{xy_2} 
+\|\lambda\|^2\,V^{(0)}_x
\,\big(V_x-V^{(0)}_x-\sqrt{\delta}\,V^{(1,0)}_x-\sqrt{\epsilon}\,V^{(0,1)}_x\big)\bigg) \\
&-\frac{\|\lambda\|^2}{2}\left(\frac{V^{(0)}_x}{V^{(0)}_{xx}}\right)^2
\frac{\big(V_{xx}-V^{(0)}_{xx}\big)^2}{V_{xx}}-\mathcal{A}^{\delta}_{y_1}\,V
-\frac{\mathcal{A}^{\epsilon}_{y_2}}{\epsilon}\Big(V-\epsilon\,V^{(2)}-\epsilon^{3/2}\,V^{(3)}
-\sqrt{\delta}\,\epsilon^{3/2}\,V^{(1,2)}\Big).
\end{split}
\end{equation*}
Here $\mathcal{A}^{\epsilon}_{y_2}\,V^{(2)}$, $\mathcal{A}^{\epsilon}_{y_2}\,V^{(3)}$ and $\mathcal{A}^{\epsilon}_{y_2}\,V^{(1,2)}$ are defined through \eqref{LV2}, \eqref{LV3} and \eqref{V1slmixed}, respectively.

\bigskip\bigskip

\end{document}